\documentclass[11pt]{article}

\usepackage[numbers,sort&compress]{natbib} 
\usepackage{fullpage}

\usepackage{xspace}
\usepackage{url}
\usepackage{setspace}
\usepackage{amsmath}
\usepackage{bbm}
\usepackage{amstext}
\usepackage{bm}
\usepackage{algorithm}
\usepackage[noend]{algorithmic}
\usepackage[section]{definitions}

\newcommand{\vc}[1]{\ensuremath{\bm{#1}}\xspace}

\newcommand{\PP}{\ensuremath{\text{P}}\xspace}
\newcommand{\NP}{\ensuremath{\text{NP}}\xspace}

\newcommand{\IS}{\textsc{Independent Set}\xspace}
\newcommand{\HEGGFULL}{\textsc{Hypergraph Edge Guessing Game}\xspace}
\newcommand{\HEGG}{\textsc{HEGG}\xspace}

\newcommand{\PWF}{Sanitized Welfare\xspace}
\newcommand{\pwf}{sanitized welfare\xspace}

\newcommand{\psu}{sanitized sender utility\xspace}

\newcommand{\FullInfWel}{full-information welfare\xspace}
\newcommand{\TotWel}{full-information welfare\xspace}
\newcommand{\ServWel}{served social welfare\xspace}
\newcommand{\FISW}{\FullInfWel}

\newcommand{\NUMSIG}{\ensuremath{M}\xspace}

\newcommand{\ALLSIGSET}{\ensuremath{\Sigma}\xspace}

\newcommand{\SIGSET}[1][]{\ensuremath{%
\ifthenelse{\equal{#1}{}}{Q}{%
\ifthenelse{\equal{#1}{'}}{Q'}{%
\ifthenelse{\equal{#1}{''}}{Q''}{%
Q_{#1}}}}}\xspace}
\newcommand{\SIGSETS}{\ensuremath{\SIGSET}\xspace}
\newcommand{\SIGSETT}{\ensuremath{\hat{\SIGSET}}\xspace}
\newcommand{\PRSET}[1][]{\ensuremath{%
\ifthenelse{\equal{#1}{}}{S}{%
\ifthenelse{\equal{#1}{'}}{S'}{%
\ifthenelse{\equal{#1}{''}}{S''}{%
\Sigma_{#1}}}}}\xspace}
\newcommand{\SIG}[1][]{\ensuremath{%
\ifthenelse{\equal{#1}{}}{\sigma}{%
\ifthenelse{\equal{#1}{'}}{{\sigma'}}{%
\ifthenelse{\equal{#1}{^S}}{{\hat{\sigma}_{S}}}{%
\ifthenelse{\equal{#1}{^T}}{{\hat{\sigma}_{T}}}{%
\sigma^{#1}}}}}}\xspace}
\newcommand{\SIGS}[1][]{\ensuremath{%
\ifthenelse{\equal{#1}{}}{\sigma}{%
\ifthenelse{\equal{#1}{'}}{{\sigma'}}{%
\ifthenelse{\equal{#1}{^S}}{{\hat{\sigma}_{S}}}{%
\ifthenelse{\equal{#1}{^T}}{{\hat{\sigma}_{T}}}{%
\sigma^{#1}}}}}}\xspace}
\newcommand{\SIGT}[1][]{\ensuremath{%
\ifthenelse{\equal{#1}{}}{\omega}{%
\ifthenelse{\equal{#1}{'}}{{\omega'}}{%
\ifthenelse{\equal{#1}{^S}}{{\hat{\omega}_{S}}}{%
\ifthenelse{\equal{#1}{^T}}{{\hat{\omega}_{T}}}{%
\omega^{#1}}}}}}\xspace}

\newcommand{\GSIG}{\ensuremath{\perp}\xspace}

\newcommand{\PostD}[1][]{\ensuremath{\ifthenelse{\equal{#1}{}}{{\mu}}{{\mu}_{#1}}}\xspace}

\newcommand{\SigScheme}[1][]{\ensuremath{%
\ifthenelse{\equal{#1}{}}{X}{%
\ifthenelse{\equal{#1}{'}}{X'}{%
\ifthenelse{\equal{#1}{I}}{\varphi}{%
X^{#1}}}}%
}\xspace}

\newcommand{\CSigP}[2]{\ensuremath{\varphi(#1,#2)}\xspace}

\newcommand{\RMap}[1]{\ensuremath{\varphi(#1)}\xspace}

\newcommand{\SigP}[3][]{\ensuremath{%
\ifthenelse{\equal{#1}{}}{x_{#2,#3}}{%
\ifthenelse{\equal{#1}{'}}{x'_{#2,#3}}{%
x^{#1}_{#2,#3}}}}\xspace}

\newcommand{\SigC}[2][]{\ensuremath{%
\ifthenelse{\equal{#1}{}}{\vc{x}_{#2}}{%
\ifthenelse{\equal{#1}{'}}{\vc{x'}_{#2}}{%
\vc{x}^{#1}_{#2}}}}\xspace}

\newcommand{\OSVec}[1][]{\ensuremath{%
\ifthenelse{\equal{#1}{}}{\vc{y}}{%
\ifthenelse{\equal{#1}{'}}{\vc{y'}}{%
\vc{y}^{#1}}}}\xspace}

\newcommand{\OS}[2][]{\ensuremath{%
\ifthenelse{\equal{#1}{}}{y_{#2}}{%
\ifthenelse{\equal{#1}{'}}{y'_{#2}}{%
y^{#1}_{#2}}}}\xspace}

\newcommand{\TS}[2][]{\ensuremath{%
\ifthenelse{\equal{#1}{}}{y_{[#2,\NUMVAL]}}{%
\ifthenelse{\equal{#1}{'}}{y'_{[#2,\NUMVAL]}}{%
y^{#1}_{[#2,\NUMVAL]}}}}\xspace}

\newcommand{\PMY}[2][]{\ensuremath{%
\ifthenelse{\equal{#1}{}}{y_{#2}}{%
\ifthenelse{\equal{#1}{'}}{y'_{#2}}{%
y^{#1}_{#2}}}}\xspace}
\newcommand{\SigI}[1]{\ensuremath{k_{#1}}\xspace}
\newcommand{\SigIS}[1]{\ensuremath{k_{#1}}\xspace}
\newcommand{\SigIT}[1]{\ensuremath{\hat{k}_{#1}}\xspace}

\newcommand{\LSig}[2]{\ensuremath{{#1}_{\leq {#2}}}\xspace}

\newcommand{\PrP}{\ensuremath{k}\xspace}

\newcommand{\NUMVAL}{\ensuremath{n}\xspace}

\newcommand{\VAL}{\ensuremath{V}\xspace}

\newcommand{\Val}[2][]{\ensuremath{%
\ifthenelse{\equal{#1}{'}}{v'_{#2}}{v_{#2}}}\xspace}

\newcommand{\VVec}[1][]{\ensuremath{%
\ifthenelse{\equal{#1}{'}}{\vc{v'}}{\vc{v}}}\xspace}

\newcommand{\DIST}{\ensuremath{\Gamma}\xspace}

\newcommand{\PVal}[2][]{\ensuremath{
\ifthenelse{\equal{#1}{'}}{p'_{#2}}{%
p_{#2}}}\xspace}

\newcommand{\PVec}[1][]{\ensuremath{%
\ifthenelse{\equal{#1}{'}}{\vc{p'}}{\vc{p}}}\xspace}

\newcommand{\RPVec}[1][]{\ensuremath{%
\ifthenelse{\equal{#1}{}}{\vc{p'}}{%
\vc{p}^{(#1)}}}\xspace}

\newcommand{\RP}[2][]{\ensuremath{%
\ifthenelse{\equal{#1}{}}{p'_{#2}}{%
p^{(#1)}_{#2}}}\xspace}

\newcommand{\PRICE}{\ensuremath{P}\xspace}

\newcommand{\NATURE}{\ensuremath{\Omega}\xspace}

\newcommand{\NState}{\ensuremath{\omega}\xspace}

\newcommand{\PNat}[1]{\ensuremath{p_{#1}}\xspace}
\newcommand{\PNatV}{\ensuremath{\vc{p}}\xspace}

\newcommand{\RECACT}{\ensuremath{A}\xspace}

\newcommand{\recact}{\ensuremath{a}\xspace}

\newcommand{\ANUTIL}[1]{\ensuremath{
\ifthenelse{\equal{#1}{S}}{u_S}{%
u_R}}\xspace}
\newcommand{\ANUtil}[3]{\ensuremath{\ANUTIL{#1}(#2,#3)}\xspace}

\newcommand{\GAME}{\ensuremath{(\NATURE, \PNatV, \RECACT, \ANUTIL{S}, \ANUTIL{R})}\xspace}

\newcommand{\Util}[2][]{\ensuremath{%
\ifthenelse{\equal{#1}{}}{U(#2)}{%
\ifthenelse{\equal{#1}{-}}{\tilde{U}(#2)}{%
U_{#1}(#2)}}
}\xspace}

\newcommand{\SUtil}[2][]{\ensuremath{%
u(#2)}\xspace}

\newcommand{\PWelfare}[2][]{\ensuremath{%
\ifthenelse{\equal{#1}{}}{\tilde{W}(#2)}{%
\tilde{W}_{#1}(#2)}}\xspace}

\newcommand{\Welfare}[2][]{\ensuremath{%
\ifthenelse{\equal{#1}{}}{W(#2)}{%
W_{#1}(#2)}}\xspace}

\newcommand{\SSW}[2]{\ensuremath{w_{#1}(#2)}\xspace}

\newcommand{\SW}[1]{\ensuremath{W(#1)}\xspace}

\newcommand{\PSW}[1]{\ensuremath{\tilde{W}(#1)}\xspace}

\newcommand{\Revenue}[3][]{\ensuremath{%
\ifthenelse{\equal{#1}{}}{R(#2,#3)}{%
R_{#1}(#2,#3)}}\xspace}

\newcommand{\Intv}[1]{\ensuremath{I_{#1}}\xspace}

\newcommand{\Rate}[1]{\ensuremath{\rho_{#1}}\xspace}

\newcommand{\PMS}{\ensuremath{X}\xspace}
\newcommand{\PMSP}{\ensuremath{X^+}\xspace}
\newcommand{\PMT}{\ensuremath{\hat{X}}\xspace}
\newcommand{\PMTP}{\ensuremath{\hat{X}^+}\xspace}

\newcommand{\inversediff}[2]{\frac{1}{#1}-\frac{1}{#2}}
\newcommand{\inversediffi}[2]{\inversediff{v_{#1}}{v_{#2}}}
\newcommand{\pmS}[2]{\ensuremath{%
\ifthenelse{\equal{#1}{}}{x}{x_{#2,#1}}}\xspace} 
\newcommand{\pmSp}[2]{\ensuremath{%
\ifthenelse{\equal{#1}{}}{x^+}{x^+_{#2,#1}}}\xspace} 
\newcommand{\pmT}[2]{\ensuremath{%
\ifthenelse{\equal{#1}{}}{\hat{x}}{\hat{x}_{#2,#1}}}\xspace}
\newcommand{\pmTp}[2]{\ensuremath{%
\ifthenelse{\equal{#1}{}}{\hat{x}^+}{\hat{x}^+_{#2,#1}}}\xspace}
\newcommand{\dS}[2]{\ensuremath{\delta^{#1}_{#2}}\xspace}
\newcommand{\dT}[2]{\ensuremath{\hat{\delta}^{#1}_{#2}}\xspace}
\newcommand{\nfS}[2][]{\ensuremath{%
\ifthenelse{\equal{#1}{}}{e_{#2}}{e^{{#1}}_{#2}}}\xspace} 
\newcommand{\nfT}[2][]{\ensuremath{%
\ifthenelse{\equal{#1}{}}{\hat{e}_{#2}}{\hat{e}^{{#1}}_{#2}}}\xspace}

\newcommand{\nfSn}[1][]{\ensuremath{%
\ifthenelse{\equal{#1}{}}{m}{m_{#1}}}}
\newcommand{\nfTn}[1][]{\ensuremath{%
\ifthenelse{\equal{#1}{}}{\hat{m}}{\hat{m}_{#1}}}}
\newcommand{\StoT}[1]{\ensuremath{\lfloor #1\rfloor}}
\newcommand{\ITtoIS}[1]{\ensuremath{\lfloor #1\rfloor}}
\newcommand{\ITtoISSUM}[1]{\SigIS{#1}\le \SigIT{\SIGT} < \SigIS{#1-1},
\SigIT{\SIGT}\le k}
\newcommand{\START}{\ensuremath{j^*}}
\newcommand{\opmS}[2]{\pmS{#1}{#2}}
\newcommand{\opmSp}[2]{\pmSp{#1}{#2}}
\newcommand{\opmT}[2]{\pmT{#1}{#2}}
\newcommand{\opmTp}[2]{\pmTp{#1}{#2}}

\DeclareMathOperator{\Rev}{Rev}

\DeclareMathOperator{\Welf}{Welfare}

\renewcommand{\hat}{\widehat}
\renewcommand{\tilde}{\widetilde}

\def\min{\qopname\relax n{min}}
\def\max{\qopname\relax n{max}}

\newcommand{\RR}{\mathbb{R}}

\def\B{\mathcal{B}}

\def\R{\mathcal{R}}

\def\eps{\epsilon}

\newcommand{\NUMSUPP}{M'}

\DeclareMathOperator{\MSW}{SW}

\newcommand{\EqualRevN}{n}
\newcommand{\MP}[2][]{\ensuremath{%
\ifthenelse{\equal{#1}{'}}{q'_{#2}}{%
\ifthenelse{\equal{#1}{2}}{q''_{#2}}{%
q_{#2}}}}\xspace}

\newcommand{\MPV}[1][]{\ensuremath{%
\ifthenelse{\equal{#1}{'}}{\vc{q'}}{%
\ifthenelse{\equal{#1}{2}}{\vc{q''}}{%
\vc{q}}}}\xspace}

\newcommand{\Bin}[1]{\ensuremath{B_{#1}}\xspace}

\providecommand{\xhdr}[1]{\subsubsection*{{#1}}}

\newcommand{\algoname}[1]{\textnormal{\textsc{#1}}}

\begin{document}

\title{Persuasion with Limited Communication}
 \author{Shaddin Dughmi \\shaddin@usc.edu \and David Kempe \\
   David.M.Kempe@gmail.com \and Ruixin Qiang\\ rqiang@usc.edu}



\maketitle
\begin{abstract}
We examine information structure design, 
also called ``persuasion'' or ``signaling,''
in the presence of a constraint on the amount of communication. 
We focus on the fundamental setting of bilateral trade, 
which in its simplest form involves a seller with a single item to price,  
a buyer whose value for the item is drawn from a common prior
distribution over $n$ different possible values, 
and a take-it-or-leave-it-offer protocol. 
A mediator with access to the buyer's type may partially reveal such
information to the seller in order to further some objective such as
the social welfare or the seller's revenue.   
%
We study how a limit on the number of bits of communication affects
this setting in two respects:
(1) How much does this constraint reduce the optimal welfare or
revenue? 
(2) What effect does constraining communication have on the
computational complexity of the mediator's optimization problem? 


In the setting of maximizing welfare under bilateral trade, 
we exhibit positive answers for both questions (1) and (2).
Whereas the optimal unconstrained scheme may involve $n$
signals (and thus $\log(n)$ bits of communication),
we show that $O(\log(n) \log \frac{1}{\eps})$ signals suffice for a
$1-\eps$ approximation to the optimal welfare, and this bound is tight.
This largely justifies the design of algorithms for signaling subject
to drastic limits on communication.
As our main result, we exhibit an efficient algorithm for computing a 
$\frac{\NUMSIG-1}{\NUMSIG} \cdot (1-1/e)$-approximation to the
welfare-maximizing scheme with at most \NUMSIG signals. 
This result hinges on an intricate submodularity argument 
which relies on the optimality of a greedy algorithm for
solving a certain linear program. 
For the revenue objective, the surprising logarithmic bound on
the number of signals does not carry over: 
we show that $\Omega(n)$ signals are needed for
a constant factor approximation to the revenue of a fully informed seller.
From a computational perspective, however, the problem gets easier:
we show that a simple dynamic program computes the 
signaling scheme with \NUMSIG signals maximizing the seller's revenue.

Observing that the signaling problem in bilateral trade is a
special case of the fundamental \emph{Bayesian Persuasion} model of
Kamenica and Gentzkow, we also examine the question of
communication-constrained signaling more generally. 
Specifically, in this model there is a sender (the mediator), a
receiver (the seller) looking to take an action (setting the price),
and a state of nature (the buyer's type) drawn from a common prior. 
The state of nature encodes both the receiver's utility and the
sender's objective as a function of the receiver's action.
Our results for bilateral trade with the revenue objective imply that
limiting communication to \NUMSIG signals can scale the sender's utility
by a factor of $O(\frac{\NUMSIG}{n})$ in general, where $n$ denotes the number
of states of nature.
We also show that our positive algorithmic results for bilateral trade do not extend
to communication-constrained signaling in the Bayesian Persuasion model. 
Specifically, we show that it is NP-hard to approximate the optimal
sender's utility to within any constant factor in the presence of
communication constraints.

%


\end{abstract}

\section{Introduction} \label{sec:introduction}

Strategic interactions are often rife with uncertainty and information
asymmetry. Auctions and markets on the Internet feature sellers with
privileged information regarding their products, and buyers with
private information regarding their willingness to pay. 
The \emph{information structures} which govern these interactions play
a key role in determining the equilibria of the resulting games.
In Akerlof's \emph{market for lemons} \cite{Akerlof70}, 
for instance, information asymmetries between the buyers and sellers of
used cars --- buyers cannot distinguish good cars from bad 
whereas sellers can --- can lead to the collapse of the entire market. 
In other cases, information asymmetries can serve a useful purpose; 
for example, a seller of advertising impressions may conflate
different impressions in order to prevent ``cherry picking'' by
advertisers, increase competition, and as a 
result increase her\footnote{For clarity, we will throughout
  use female pronouns for the seller and male pronouns for the buyer.}
revenue \cite{levinmilgrom}. 
It is for these reasons that information structure design, also known
as \emph{signaling}, is emerging as a new \emph{mechanism design for
  information}. This new frontier, like traditional mechanism design,
raises deep algorithmic and complexity-theoretic questions whose
exploration has only recently begun 
(see, e.g., \cite{emek:feldman:gamzu:paes-leme:tennenholtz,bro-miltersen:sheffet:mixed,Guo13,dughmi:immorlica:roth:auction,dughmi:signaling-hardness,DX16}).


Perhaps one of the most fundamental economic interactions
governed by the presence or absence of information is 
\emph{bilateral trade} between two parties \cite{mas-collel:whinston:green}. 
In (a simplified form of) bilateral trade, 
one side can choose whether to participate in the trade, 
and by doing so would generate a social surplus which is
private information to him.
The other side can propose to take a fixed amount of the social surplus.
Two particularly natural instantiations of bilateral trade are the
following: 
\begin{enumerate}
\item Trade of an item between a seller (who has no value for the item) 
and a buyer via a posted price. 
The seller's posted price is the amount of surplus she proposes to take, 
while the buyer's valuation for the item, 
drawn from a commonly known distribution, 
is the amount of social surplus the trade would generate. 
The buyer chooses whether to accept the seller's posted price.
We call the resulting game the \emph{pricing game}.
\item Trade between an employer and an agent:\footnote{In the
literature, this falls into the class of \emph{principal-agent} models. 
However, in this article, we use the word ``principal'' for a
different role, so we will use the non-standard nomenclature to avoid
misunderstandings.}
an employer would like to hire an agent to complete a project, 
and has (known) utility $u$ for its completion.
The agent has a private cost $c$ (drawn from a known distribution) for
completing the project, 
and the social surplus generated is the difference $u-c$.
Without knowing the cost of the agent, the employer posts a
proposed payment $p$, which is equivalent to posting the share $u-p$
of the social welfare which she proposes to keep.
We call this game the \emph{employment game}.
\end{enumerate}

The buyer/agent is assumed to be rational with quasilinear utility, 
while the seller/employer aims to maximize her own utility.
The buyer/agent accepts an offer if he would derive non-negative
utility from it, and rejects it otherwise; this corresponds to the
valuation exceeding the price in the pricing game, and the payment
exceeding the cost in the employment game.
For concreteness, we will state all of our results in the language of
the pricing game; however, all results carry over verbatim to the
employment game, and we will occasionally remark on the interpretation
of results in this context.

The seller's chosen offer price to the buyer depends on what she knows
about the buyer's value.
If the seller has no information other than the distribution
\DIST of values in the population of potential buyers, 
she chooses the price $\PRICE^* = \PRICE^*(\DIST)$ maximizing her revenue 
$\Rev(\PRICE,\DIST) = \PRICE \cdot \Prob[\VAL \sim \DIST]{\VAL \geq \PRICE}$. 
This leads to a revenue of
$\Rev(\DIST) = \Rev(\PRICE^*(\DIST),\DIST) = \PRICE^* \cdot \Prob[\VAL \sim \DIST]{\VAL \geq \PRICE^*}$
for the seller,
and a social welfare of 
$\Welf(\DIST) = \Prob[\DIST]{\VAL \geq \PRICE^*} 
          \cdot \ExpectC[\VAL \sim \DIST]{\VAL}{\VAL \geq \PRICE^*}$ 
for both players combined. 
At the other extreme is the case when the seller is fully
informed about $\VAL$. She can now set a price $\PRICE = \VAL$; 
trade will always occur, leading to a maximum social welfare of 
$\Expect[\VAL \sim \DIST]{\VAL}$, which is fully extracted as
revenue by the seller.
Notice the difference caused by different amounts of information being
communicated: with no information, the social welfare can be
arbitrarily smaller than with full information.


In order to fully inform the seller, very fine-grained information had
to be communicated. In reality, for practical and logistical
reasons, the information received by the seller
about the buyer is typically limited. 
For example, in the employment game, the information
may be provided by a university or certification agency, which 
may initially only be able to communicate a coarse-grained
rating of the agent via a GPA or the performance on a certification
test.
When the ``agent'' provides a product (e.g., a piece of
  clothing or medication), the location or label it is sold under
  (boutique/brand-name vs.~discount/generic) sends a coarse signal to
  the potential employer/buyer about the distribution of qualities she
  is to expect.
Indeed, the study of communication constraints and their
impact on the outcomes of games dates back at least to the work of
\cite{blumrosen2006implementation,blumrosen2007auctions} on
communication constraints in auctions.

\emph{How to optimally inform the seller with limited communication is
the subject of the present article.}\footnote{
We note that an alternate interpretation of this goal is in terms of
\emph{market segmentation}. 
Specifically, \emph{how would a market designer optimally partition
the market into a limited number of segments?}}
This question actually comprises two separate thrusts:
(1) What is the inherent \emph{price of limited communication}?
In other words, how much social welfare is lost
because the principal can only communicate limited information to the seller?
A particularly stark version of this question is the following:
in the presence of a self-interested seller, how much social welfare
can a principal salvage by sending a single bit of information,
which is equivalent to merely being able to \emph{exclude} some buyers from
the market? 
(2) What are the \emph{algorithmic} consequences of limited communication? 
How well can the principal optimize the welfare with limited
communication, if the computation has to be efficient as well?



In order to formalize the notion of limited communication, 
we first define information structures.
At its most general, an \emph{information structure} for the seller
is a (possibly randomized) map \SigScheme[I] from the realized value 
$\VAL \in \RR$ of the buyer to a signal $\SIG \in \ALLSIGSET$ presented
to the seller. 
This in effect partitions the probability histogram of \DIST
--- or equivalently, the population of buyers --- into different
segments, with each corresponding to a signal. 
Specifically, we write
$\CSigP{v}{\SIG} = \ProbC{\RMap{v} = \SIG}{\VAL = v}$, 
where the randomness is over the internal coins of \SigScheme[I].
Receiving signal \SIG induces, via Bayes' rule, a posterior
distribution \PostD[\SIG] for the seller, characterized by
$\Prob[{\PostD[\SIG]}]{\VAL = v} 
= \frac{\Prob[\DIST]{\VAL = v} \cdot \CSigP{v}{\SIG}}{\Prob{\RMap{\VAL} = \SIG}}$;
here, the randomness in the denominator is over both \DIST and the
internal coins of \SigScheme[I].
Upon receiving \SIG, the seller's optimal price is $\PRICE^*(\PostD[\SIG])$
maximizing $\Rev(\PRICE,\PostD[\SIG])$.
This price induces a revenue of $\Rev(\PostD[\SIG])$ for the seller
and a social welfare of $\Welf(\PostD[\SIG])$. 
The expected revenue 
and social welfare over all draws of the buyer's value and randomness
in the scheme \SigScheme[I] are then given by 
$\Rev(\SigScheme[I],\DIST) 
= \sum_{\SIG \in \ALLSIGSET} \Prob{\RMap{\VAL} = \SIG} \cdot \Rev(\PostD[\SIG])$  
and $\Welf(\SigScheme[I],\DIST) 
= \sum_{\SIG \in \ALLSIGSET} \Prob{\RMap{\VAL} = \SIG} \cdot \Welf(\PostD[\SIG])$,
respectively. 

We now revisit our motivating examples.
Suppose that the distribution \DIST has support 
$\SET{\Val{1}, \Val{2}, \ldots, \Val{\NUMVAL}}$.
Then, \VAL can be precisely communicated to the seller by choosing
a signal set $\ALLSIGSET = \SET{1, \ldots, \NUMVAL}$
and setting $\CSigP{\Val{i}}{i} = 1$ (and $\CSigP{\Val{i}}{j} = 0$ for
$i \neq j$).
By way of contrast, the signaling scheme communicating no information
to the seller is implemented with $\ALLSIGSET = \SET{1}$
and $\CSigP{\Val{i}}{1} = 1$ for all $i$.
Notice that the latter uses much lower ``communication complexity,''
as measured by $\SetCard{\ALLSIGSET}$. 
Also notice that both these extreme information structures are inherently
algorithmically efficient:
given an explicit representation of \DIST and a value \VAL, 
it is trivial to compute \RMap{\VAL}. 
For the design of signaling schemes, we use the constraint that
$\SetCard{\ALLSIGSET} = \NUMSIG$, for a given bound \NUMSIG,
as the ``low communication complexity'' constraint.

We would like to understand the impact of the communication
complexity on the social welfare  that can be (in principle) achieved, 
and to analyze the algorithmic question of computing optimal signaling 
schemes \SigScheme[I] with limited communication complexity.
The ``gold standard'' for social welfare is 
$\Expect[\VAL \sim \DIST]{\VAL}$, which we will call the
 \FullInfWel.

Formally, we are interested in the following questions: 
Given an explicit representation of the distribution \DIST
and a bound \NUMSIG on the number of available signals, 
(1) What fraction of the \TotWel can be obtained by a signaling scheme
with \NUMSIG signals?
(2) Are there (approximately) optimal and computationally efficient
signaling schemes for maximizing social welfare?

\subsection*{Our Results}

Our first main result (proved in Section~\ref{sec:single-segment})
shows that signaling schemes even with an extremely limited number of
signals are surprisingly powerful in extracting welfare.

\begin{theorem} \label{thm:single-signal-welfare}
For any distribution \DIST with support size $n$, 
there is a single segment of the market 
(i.e., a non-negative vector indexed by buyer types that is
  pointwise upper-bounded by the prior over buyer types)
with social welfare at least a $\Theta(1/\log n)$ 
fraction of the \TotWel $\Expect[\VAL \sim \DIST]{\VAL}$.
The $\Theta(1/\log n)$ bound is tight.
\end{theorem}%

Notice that this result is quite surprising.
There are value distributions (such as equal-revenue distributions)
under which the presence of a self-interested seller results in only a
fraction $1/n$ of the \TotWel being realized.
The theorem states that merely by \emph{excluding} some buyers,
the principal can improve this bound to $\Theta(1/\log n)$.
Perhaps a ``natural'' conjecture would have been that using \NUMSIG
segments, at most a fraction $O(\NUMSIG/n)$ of welfare could be
attained in the worst case. 
The theorem and the subsequent corollary show that this conjecture is
false. However, as we will see shortly, it is in fact true for the
objective of maximizing the seller's \emph{revenue}.

Applying Theorem \ref{thm:single-signal-welfare} repeatedly yields
Corollary \ref{cor:logn-signal-welfare}, which 
shows that with a relatively small number of signals, we can
get arbitrarily close to the \TotWel.

\begin{corollary} \label{cor:logn-signal-welfare}
For any $\eps >0$ and any distribution \DIST with support size $n$, 
there is a signaling scheme with $O(\log n \log(1/\eps))$ signals
which obtains a $(1-\eps)$ fraction of the \FullInfWel.
\end{corollary}

In terms of the number of bits of communication required,
Corollary~\ref{cor:logn-signal-welfare} implies that communicating 
$O(\log \log n + \log \log(1/\eps))$ bits extracts a $(1-\eps)$
fraction of the social welfare that could be extracted using $\log n$ bits.

Corollary~\ref{cor:logn-signal-welfare} also has the following
algorithmic implication: exhaustively searching over all 
signaling schemes with $\Theta(\log n \cdot \log (1/\epsilon))$
signals, one obtains a QPTAS (quasi-polynomial time approximation scheme). 

Obtaining a truly polynomial-time algorithm appears quite a
bit more challenging, although --- as discussed in
Section~\ref{sec:conclusion} --- we currently do not have any
hardness results, even for exact optimization. 
Our main technical result is Theorem~\ref{thm:welfare-intro}, 
which shows that one can obtain a constant factor approximation to the
social welfare in polynomial time.

\begin{theorem}\label{thm:welfare-intro}
For any $\NUMSIG>1$, there is a polynomial-time 
$\frac{\NUMSIG-1}{\NUMSIG} \cdot (1-1/e)$ approximation algorithm 
for the problem of implementing a welfare-maximizing signaling scheme
with at most \NUMSIG signals, given an explicit representation of \DIST.
\end{theorem}

The proof of this theorem is quite involved.
At the heart of it is a proof that the social welfare achieved from a set of
signals is submodular.
More precisely, the proof focuses on the social welfare accrued from
all but one signal 
(which we denote by \GSIG and call the \emph{garbage signal} 
--- it is a signal of buyer types for which our solution will
  not be credited with any reward).
Each non-garbage signal \SIG induces an optimal equilibrium price
$\PRICE^*(\PostD[\SIG])$ chosen by the seller.
We can think of the computation of \SigScheme[I] as first choosing the
set $S = \SET{\PRICE_1, \PRICE_2, \ldots, \PRICE_{\NUMSIG-1}}$ of equilibrium
prices of the optimum solution, and then choosing an optimal signaling
scheme inducing these particular prices.
The key insight is that the optimum social welfare with price set $S$
(and one garbage signal) is a monotone and submodular function of $S$.
This fact is proved as Theorem~\ref{thm:submodularity} in
Section~\ref{sec:submodularity}.
The proof relies heavily on a characterization of the optimal
signaling scheme inducing the set $S$ of prices.
While it is easy to see that, given the target prices
$S = \SET{\PRICE_1, \PRICE_2, \ldots, \PRICE_{\NUMSIG-1}}$,
this optimal signaling scheme can be
computed using a linear program (see Section~\ref{sec:preliminaries}),
we require a better characterization, and thereto show 
(as Theorem~\ref{thm:greedy-optimal} in Section~\ref{sec:greedy}) 
that it is also the output of a greedy algorithm.

\subsubsection*{Optimizing Seller Revenue}

The reader may have noticed our focus on social welfare.
Almost equally frequently studied is the objective of maximizing
seller revenue.
It turns out that results for seller revenue are much more
straightforward, both technically and in terms of their implications.
First, the communication constraint can severely curtail the seller's
revenue: when the buyer's value is drawn from an
equal-revenue distribution supported on a geometric
progression of length $n$, the revenue-maximizing signaling scheme
with \NUMSIG signals recovers only an $O(\NUMSIG/n)$ fraction of
the full-information revenue.
On the other hand, \emph{computing} the optimal signaling scheme for
revenue maximization is straightforward (see Section~\ref{sec:seller-revenue}):

\begin{theorem}\label{thm:revenue-intro}
The optimal signaling scheme for maximizing seller revenue
  groups buyers into \NUMSIG contiguous segments by valuations, i.e., if
  the same signal \SIG is sent for buyers with valuations $v < v'$, 
  then \SIG is also sent for all buyers with valuations $v'' \in [v,v']$.
As a result, there exists a polynomial-time dynamic
programming algorithm which, given an explicit representation of \DIST
and a bound \NUMSIG on the number of signals, 
computes a signaling scheme maximizing the seller's expected revenue.
\end{theorem}

In the employment game, Theorem~\ref{thm:revenue-intro} confirms
(for the purpose of employer utility maximization) the generally
agreed-upon form of grading or performance evaluation, wherein the
highest performers are grouped together in one category (`A'),
followed by the next highest category (`B'), etc.
Such signaling schemes are generally not optimal if the
goal is to maximize social welfare, and indeed, the optimum
signaling scheme for welfare maximization can be fairly complex.
As an example, consider a buyer distribution supported on types
$(1, 2, 3, 4)$ with probabilities
$(2/12, 1/12, 2/12, 7/12)$, respectively.
When $\NUMSIG=2$ signals are allowed, the optimal signaling scheme
obtains full welfare by sending signals with posterior unnormalized
probabilities of $(2/12, 1/12, 0, 1/12)$ and $(0, 0, 2/12, 6/12)$.


\subsubsection*{Bayesian Persuasion}

While our main focus is on bilateral trade, the
framework of communication-bounded signaling naturally applies to
other games as well. 
A more general setting has been termed \emph{Bayesian Persuasion} by
\citet{Kamenica11}: a sender observes a random variable
capturing the ``state of the world,'' and can send a signal to a
receiver. 
The receiver, based on the received signal, chooses an action. 
The utility of both the sender and the receiver depend on the state of
the world and the chosen action, and are not necessarily aligned. 
Thus, the sender's goal is to design the information structure
such that the receiver will choose actions which are in expectation
beneficial to him.

Notice that signaling in bilateral trade fits in this framework. 
In the pricing game, the seller is the receiver, the buyer's valuation
is the state of the world, and the sender is a market designer with
the goal of maximizing the seller's revenue or social welfare. 
In the employment game, the employer is the receiver, the
employee's cost is the state of the world, and the sender is an
educational institution or crowdsourcing 
website aiming to generate welfare for its participants.

Our results for revenue in the pricing game already imply that the
price of limited communication is high in some persuasion games. 
While it may be natural to conjecture that the submodularity property 
carries over from bilateral trade to more general
persuasion games, this is not the case:
we establish a strong hardness result for maximizing the
sender's utility in general.
(A more formal version of this theorem and the proof are given in
Section~\ref{sec:hardness}.)

\begin{theorem} \label{thm:hardness-intro}
For any constant $c > 0$, it is \NP-hard to construct a signaling scheme
approximating the maximum expected sender utility to within a factor
$c$, given an explicit representation of a Bayesian Persuasion game
and a bound \NUMSIG on the number of signals.
\end{theorem}

\subsection*{Related Work}


Our focus is on the classical model of \emph{bilateral trade} 
(see \cite[Chapter 23]{mas-collel:whinston:green}).
Our choice of protocol, namely the take-it-or-leave-it offer, 
is arguably the simplest mechanism for bilateral trade, 
and in the case of the pricing game corresponds to the revenue-optimal
mechanism by the classical result of \citet{myerson:optimal-auction}.
The study of the impact of auxiliary information on trade --- also
known as \emph{third degree price discrimination} --- 
has a long history, starting at least as early as \cite{pigou:welfare}.
We refer the reader to \cite{Bergemann14} for an in-depth look at this
economic literature.

The work most directly related to ours is that of \citet{Bergemann14},
who examine the effects of information in the same pricing game. 
Their main result is a remarkable characterization
of buyer and seller expected utilities that are attainable by varying
the information structure of the seller, i.e., by segmenting the
market and allowing the seller to \emph{price discriminate} between
segments. 
They characterize the space of \emph{realizable} pairs $(r,u)$ for
which there exists an information structure \SigScheme[I] 
such that the seller's expected revenue is $r$ and the buyer's
expected utility is $u$:
$(r,u)$ is realizable if and only if $r \geq \Rev(\DIST)$ 
(i.e., the seller at least matches her ``uninformed'' revenue),
$ u \geq 0$, and $r + u \leq \Expect[\VAL \sim \DIST]{\VAL}$.

%

Implicit in \cite{Bergemann14} is a family of
algorithms --- parametrized by the distribution \DIST and a realizable
pair of utilities $(r,u)$ --- which \emph{implement} a signaling
scheme \SigScheme[I] realizing the pair of utilities $(r,u)$.  
When \DIST is an explicitly-described distribution with support size
\NUMVAL, the signaling schemes implicit in \cite{Bergemann14} are
efficient; their runtime is a low-order polynomial in
$n$.\footnote{A particularly beautiful
example of the schemes implicit in \cite{Bergemann14} is the greedy
algorithm achieving $r = \Rev(\DIST)$ and 
$u = \Expect[\VAL \sim \DIST]{\VAL} - \Rev(\DIST)$.}
However, the most interesting signaling schemes implied in
\cite{Bergemann14} --- in particular those with largest and smallest $u$
---  use as many signals as the support size of the buyer distribution
\DIST.
This realization  motivates our examination of schemes with
limited communication.

\citet{roesler2016buyer} also study the impact of
information revelation on bilateral trade. In their model, the
buyer can observe a signal of his value for the item, and will pay
for it if the conditional expected value is weakly larger than the
price. The seller will choose the optimal monopoly price according
to the buyer's information structure. 

The signaling problem in bilateral trade
is a special case of \emph{Bayesian Persuasion}, as formalized by
\citet{Kamenica11}, generalizing an earlier model
by \citet{Brocas07}. 
%
Instantiations, variants, and generalizations of the Bayesian
Persuasion problem have seen a flurry of interest in recent years. 
For example, persuasion has been examined in the context of
voting~\cite{Alonso14}, security~\cite{Xu15,Rabinovich15},
multi-armed bandits~\cite{crowds,bandit_exploration}, 
medical research~\cite{kolotilin}, 
and financial regulation~\cite{stresstest1,stresstest2}.
\citet{DX16} also consider persuasion algorithmically. 
However, they focus on Bayesian Persuasion without any communication
constraint, but allowing exponentially (or infinitely) many states of
nature in the number of actions. 

More generally, Bayesian Persuasion is a special case of
optimal information structure design in games. 
Recent work in computer science has examined this
question algorithmically, mostly in the context of
auctions~\cite{emek:feldman:gamzu:paes-leme:tennenholtz,bro-miltersen:sheffet:mixed,Guo13,dughmi:immorlica:roth:auction,DIOT15}.
In all these works, the uncertainty (i.e., state of nature) concerns
the item being sold, rather than the type of the buyer as in our model. 
Nevertheless, \cite{DIOT15,dughmi:immorlica:roth:auction} are related to our
work in that they also examine communication-limited signaling schemes. 
The work of \citet{dughmi:signaling-hardness} 
examines the complexity of signaling in abstract two-player
normal form games, while the recent work of \citet{mixture_selection}
presents an algorithmic framework for tackling a number of
(unconstrained) signaling problems.

An analogy can be drawn between our work and some of the work
on auction design subject to communication constraints. 
\citet{blumrosen2007auctions} study single-item auctions in which
bidders can only communicate a limited number of bits to the auctioneer. 
They show that even severe bounds on communication only lead
to mild losses in welfare and revenue. Moreover, they show that bidders
simply report an interval in which their value for the item lies when
faced with an optimal auction. 
\citet{blumrosen2006implementation} study communication-constrained
mechanism design in single-parameter problems more generally, 
and examine necessary and sufficient conditions under which a
communication-constrained mechanism matches or approximates the
optimal (unconstrained) mechanism.





\section{Preliminaries}
\label{sec:preliminaries}

Throughout, we use the following conventions for notation.
Vectors are denoted by bold face. When we write
$\vc{x} \leq \vc{y}$ for vectors $\vc{x}, \vc{y}$, we mean that
$x_i \leq y_i$ for all $i$.
We will frequently want to reason about the sums of entries of a
vector over a given set of indices.
We then write $x_I = \sum_{i \in I} x_i$.
We also apply this notation for elements of a matrix
$X = (x_{i,j})_{i,j}$, writing 
$x_{I,J} = \sum_{i \in I} \sum_{j \in J} x_{i,j}$.
We will particularly use this notation when $I, J$ are 
(closed or half-open) intervals of integers.

\subsection{Signaling Schemes}
When constructing a signaling scheme, we assume that the distribution
\DIST of buyer valuations is given explicitly as input. 
In particular, this means that it must have finite support of size \NUMVAL.
We assume that it is given by the valuations
$\Val{1} < \Val{2} < \ldots < \Val{\NUMVAL}$,
and their associated probabilities
$\PVal{1}, \PVal{2}, \ldots, \PVal{\NUMVAL}$, satisfying
$\sum_i \PVal{i} = 1$.
We write \VVec and \PVec for the vectors of all these values and
probabilities, respectively.

In the introduction, for ease of exposition, we described a signaling
scheme \SigScheme[I] in terms of its conditional probabilities
$\CSigP{v}{\SIG} = \ProbC{\RMap{v}=\SIG}{\VAL=v}$.
For the remainder of this article, we use notation differing in two
ways: (1) since all values are of the form \Val{i}, we can index the
buyer types by $i$ instead of $v$, and (2) it is much more convenient
to use unnormalized probabilities instead of conditional
probabilities: 
$\SigP{i}{\SIG} = \CSigP{\Val{i}}{\SIG} \cdot \Prob[\DIST]{\VAL = \Val{i}}$ is the
probability that the buyer's valuation is \Val{i} and the signal \SIG is sent.
The signaling scheme is then fully described by the matrix 
$\SigScheme \in [0,1]^{\NUMVAL \times \NUMSIG} =
(\SigP{i}{\SIG})_{i \in \SET{1, \ldots, \NUMVAL}, \SIG \in \SET{1, \ldots, \NUMSIG}}$,
satisfying that $\sum_{\SIG} \SigP{i}{\SIG} = \PVal{i}$.
From now on, we will therefore simply refer to the signaling scheme as
\SigScheme instead of \SigScheme[I].  

We sometimes describe a signal \SIG in isolation by a nonnegative
type-indexed vector $\vc{0} \leq \SigC[']{\sigma} \leq \PVec$.
We call such a vector a \emph{segment} of the market \PVec. 
Thus, a signaling scheme can be thought of as a family of segments,
one per signal, whose sum is the entire market \PVec. 

As discussed in the introduction, upon receiving the signal \SIG, the
seller will choose a price $\PRICE(\SIG)$ maximizing 
$\PRICE \cdot \Prob[{\VAL \sim \PostD[\SIG]}]{\VAL \geq \PRICE}$.
This price will always be one of the possible buyer valuations
\Val{i}, as any other price could be raised slightly without losing
any buyers. 
Furthermore, by merging signals with the same price into
one signal, without loss of generality, there are no two signals for
which the seller chooses the same price \cite{Kamenica11}. 
Hence, any signaling scheme \SigScheme induces indices 
$\SigI{1}, \SigI{2}, \ldots, \SigI{\NUMSIG}$ such that upon receiving
signal \SIG, the seller chooses price \Val{\SigI{\SIG}}.
Without loss of generality, we can rearrange the signals so that
$\SigI{1} > \SigI{2} > \ldots > \SigI{\NUMSIG}$.

For any signal \SIG, the expected welfare resulting from \SIG under
\SigScheme is called the \emph{\ServWel} and defined as
$\SSW{\SigScheme}{\SIG} = \sum_{i \geq \SigI{\SIG}} \Val{i} \cdot \SigP{i}{\SIG}$;
the social welfare is then
$\SW{\SigScheme} = \sum_{\SIG} \SSW{\SigScheme}{\SIG}$. 

We call the singaling scheme $\SigScheme$ \emph{optimal for
welfare} if $\SigScheme$ maximizes $\SW{\SigScheme}$. 
If $\SW{\SigScheme} \geq \alpha \SW{\SigScheme^\star}$, 
where $\SigScheme^\star$ is optimal for welfare, 
we call $\SigScheme$ an \emph{$\alpha$-approximation signaling scheme for
welfare}.

\subsection{Welfare and \PWF}
While our goal for most of this article is to maximize the social
welfare, it turns out that a slightly modified objective function is
significantly more amenable to analysis, both in terms of positive
results for bilateral trade and the hardness result for more general
persuasion. 
Specifically, there is a designated \emph{garbage signal} \GSIG,
and any welfare accrued when \GSIG is sent is discounted:
we thus define the \emph{\pwf} of \SigScheme to be
$\PSW{\SigScheme} = \sum_{\SIG \neq \GSIG} \SSW{\SigScheme}{\SIG}$.
By designating the signal \SIG minimizing \SSW{\SigScheme}{\SIG}
as the garbage signal, we observe:

\begin{proposition} \label{prop:welfare-sanitized}
For all signaling schemes \SigScheme, we have that
$\frac{\NUMSIG-1}{\NUMSIG} \cdot \SW{\SigScheme}
\leq \PSW{\SigScheme} \leq \SW{\SigScheme}$.
\end{proposition}

In particular, any signaling scheme \SigScheme maximizing \pwf to
within a factor $\alpha$ also maximizes social welfare to within a
factor $\frac{\NUMSIG-1}{\NUMSIG} \cdot \alpha$.
Since we will always focus on \pwf in the context of welfare
maximization, to avoid having to write $\NUMSIG-1$ for the number of
signals everywhere, we will explicitly assume that our signaling
schemes can use \NUMSIG signals \emph{in addition to} the garbage
signal \GSIG.

\subsection{\PWF Maximization Given Price Points}
Conceptually, the task of designing a good signaling scheme can be
divided into two steps:
(1) Choose the seller's price points 
$\SigI{1} > \SigI{2} > \cdots > \SigI{\NUMSIG}$
for non-garbage signals;
(2) Design a signaling scheme \SigScheme maximizing \PSW{\SigScheme}
such that the seller's best response to each signal \SIG is in fact
\SigI{\SIG}.

Given the chosen price point indices
$\SigI{1}, \ldots, \SigI{\NUMSIG}$, the goal of finding the
welfare-maximizing signaling scheme is characterized by the following
linear program.

\begin{LP}[lp:welfare-maximization]{Maximize}{%
\sum_{\SIG=1}^\NUMSIG \sum_{i=\SigI{\SIG}}^\NUMVAL \Val{i} \SigP{i}{\SIG}}

\sum_{\SIG=1}^\NUMSIG \SigP{i}{\SIG} \leq \PVal{i}
     & \text{ for all } i \quad \text{ (probability)}\\
\Val{\SigI{\SIG}} \cdot \sum_{i \geq \SigI{\SIG}} \SigP{i}{\SIG} 
   \geq \Val{k} \cdot \sum_{i \geq k} \SigP{i}{\SIG}
     & \text { for all } \SIG, k
       \quad \text{ (revenue)} \\
\SigP{i}{\SIG} \geq 0 
     & \text { for all } i, \SIG.
\end{LP}

The probability constraints capture that we do indeed have a valid
signaling scheme (with the garbage signal being assigned all residual
probabilities $\PVal{i} - \sum_{\SIG=1}^\NUMSIG \SigP{i}{\SIG}$).
The revenue constraints capture that it is a best response for the seller
to set the price point \SigI{\SIG} when receiving the signal \SIG.

Notice that for step (2), the LP~\eqref{lp:welfare-maximization}
actually achieves the optimal solution for given price points. 
The approximation is necessary because for step (1), choosing the
optimal price points appears more difficult.

\subsection{Bayesian Persuasion}
Bayesian Persuasion \cite{Kamenica11} is a natural and strong generalization
of signaling in bilateral trade.
The game involves a sender and a receiver, and is characterized by the following:
(1) a distribution over states of nature $\NState \in \NATURE$,
(2) a set \RECACT of actions the receiver can choose from, and
(3) utility functions \ANUtil{S}{\NState}{\recact}, \ANUtil{R}{\NState}{\recact}
mapping the state of nature and action chosen by the receiver to the
utilities obtained by the two players.

To illustrate these abstract concepts, for the pricing game,
the state of nature was the buyer's valuation,
and the receiver's actions were prices \PRICE.
In the present work, we assume that the distribution over states of
nature is  given explicitly by the vector \PNatV of probabilities \PNat{\NState}.

As with the pricing game, the sender can choose a signaling scheme
$\SigScheme = (\SigP{\NState}{\SIG})_{\NState \in \NATURE, \SIG \in \ALLSIGSET}$ 
with \NUMSIG signals. 
Upon receiving the signal \SIG, the receiver will choose an action
$\recact(\SIG) \in \RECACT$ maximizing
$\Expect[{\NState \sim \PostD[\SIG]}]{\ANUtil{R}{\NState}{\recact}}$,
breaking ties in favor of the sender; 
here, \PostD[\SIG] denotes the posterior distribution over states of
nature conditioned on receiving signal \SIG.

The sender's utility from sending the signal \SIG is
$\SUtil{\SIG} = \sum_{\NState} \SigP{\NState}{\SIG} \cdot \ANUtil{S}{\NState}{\recact(\SIG)}$,
and he would like to maximize his overall expected utility
$\Util{\SigScheme} = \sum_{\SIG} \SUtil{\SIG}$.
As with the special case of social welfare, we will mostly focus on
the \psu 
$\Util[-]{\SigScheme} = \sum_{\SIG \neq \GSIG} \SUtil{\SIG}$
of signaling schemes with a dedicated garbage signal \GSIG;
by the same argument as for social welfare, the \psu is within at least
a factor $\frac{\NUMSIG-1}{\NUMSIG}$ of the sender utility.


\section{Welfare with Limited Communication}
\label{sec:single-segment}

In this section, we provide a proof of Theorem~\ref{thm:single-signal-welfare}.
We then discuss some of its implications, including
Corollary~\ref{cor:logn-signal-welfare} and a QPTAS for the problem of
maximizing social welfare subject to limited communication.  

We begin by showing a special case of the theorem when each $\Val{i}$
is a power of 2. 
We then  reduce the general case to this
special case at a loss of a constant factor.

\begin{lemma} \label{lem:power-2-single}
If each $\Val{i}$ is a power of $2$, then the ratio of the \FISW to
the maximum \ServWel of a single segment of the market 
is at most $O(\log n)$.
\end{lemma}

\begin{proof}
The lemma is trivial for $n=1$, so assume that $n \geq 2$.
Let $\MSW=\sum_i \PVal{i} \Val{i}$ be the \FISW. 
First, we remove (i.e., exclude from the segment we
construct) all types $i$ with 
$\PVal{i} \Val{i} < \frac{\MSW}{n^2}$.
Because there are at most $n$ such types, this can decrease the \FISW
by at most a factor of $1-1/n$. 

Now, we group all (remaining) types $i$ into $O(\log{n})$ bins according 
to $\PVal{i}\Val{i}$.
Specifically, bin $\Bin{j}, j \geq 0$ contains all types $i$ with
$\PVal{i}\Val{i} \in (\frac{\MSW}{2^{j+1}}, \frac{\MSW}{2^{j}}]$. 
Since $\PVal{i}\Val{i} \geq \frac{\MSW}{n^2}$,
there are at most $2\log{n}$ bins. 
Thus, there is at least one bin $\Bin{j}$ such that 
$\sum_{i \in \Bin{j}} \PVal{i}\Val{i} \geq \frac{n-1}{n} \cdot \frac{1}{2 \log n} \cdot \MSW$.
Fix such a $j$ for the rest of the proof, and define $u=\frac{\MSW}{2^{j+1}}$.

Let $i^* = \min \Bin{j}$ be the type in \Bin{j} corresponding to the smallest value.
We define a market segment \MPV as follows: 
Let $\MP{i^*} = \frac{u}{\Val{i^*}}$, 
let $\MP{i} = \frac{u}{2\Val{i}}$ for all $i \in \Bin{j}$ with $i \neq i^*$, 
and let $\MP{i}=0$ for $i \not\in B_j$.

By definition of the bins, 
for all $i \in B_j$ we have $u < \PVal{i}\Val{i} \leq 2u$,
and therefore $\PVal{i}/4 \leq \MP{i} \leq \PVal{i}$.
In particular, $\vc{0} \leq \MPV \leq \PVec$,
and therefore \MPV describes a valid segment of \PVec.

We next show that, when facing the market segment \MPV,
the seller will choose the price point $i^*$. 
The seller's revenue from choosing $i^*$ is at least
$\MP{i^\ast} \Val{i^\ast} = u$.
On the other hand, for any $i > i^*$, the seller obtains revenue at most
$\Val{i} \cdot \sum_{k \geq i, k \in \Bin{j}} \frac{u}{2\Val{k}}
\leq \Val{i} \cdot \frac{u}{2\Val{i}} \cdot \sum_{k=0}^\infty \frac{1}{2^k} = u$;
here, we used that the valuations are powers of 2.
Therefore, $\Val{i^*}$ is the revenue-maximizing price for \MPV.
With the seller choosing $\Val{i^*}$ as the price, the \ServWel of \MPV is 
\begin{align*}
\sum_{i \in \Bin{j}} \MP{i} \cdot \Val{i} 
& \geq \quarter \sum_{i \in \Bin{j}} \PVal{i} \cdot \Val{i}
\; \geq \; \frac{(n-1)}{8 n\log n} \cdot \MSW
\; \geq \; \frac{1}{16 n \log n} \cdot \MSW.
\end{align*}
%
Thus, \MPV is a segment with \ServWel at least 
$\frac{1}{O(\log n)} \cdot \MSW$.
\end{proof}

We utilize Lemma~\ref{lem:power-2-single} to prove the upper-bound
portion of Theorem~\ref{thm:single-signal-welfare}.
The idea is to round down all valuations to the nearest power of 2
(losing at most a factor of 2 in the welfare), then utilize 
Lemma~\ref{lem:power-2-single} to construct a segment for the new
valuation distribution achieving at least a factor $\Omega(1/\log n)$
of the \FISW, and finally reconstruct a segment/signal for
the original distribution.
For the matching lower bound, we use an ``equal welfare distribution.''

\begin{extraproof}{Theorem~\ref{thm:single-signal-welfare}}
Let $\Val[']{j}, 1 \leq j \leq m$, be the valuations 
of the original distribution,
with associated probabilities $\PVal[']{j}$.
The new valuations are $\Val{i} = 2^i$, with
$\PVal{i} = \sum_{2^i \leq \Val[']{j} < 2^{i+1}} \PVal[']{j}$. 
We allow the index $i$ to be negative for notational convenience.
Let $n$ denote the support size of \PVec, and note that $n \leq m$. 
Because valuations were reduced by at most a factor of 2, the \FISW
of the new distribution is at least half that of the original one, i.e., 
$\sum_{i} \PVal{i} \Val{i} \geq \half \sum_{j} \PVal[']{j} \Val[']{j}$. 

Let \MPV be a single segment with \ServWel
at least a $1/O(\log n)$ fraction of the \FISW of \PVec.
By Lemma~\ref{lem:power-2-single}, such a \MPV exists.
Let $i^*$ be the price point chosen by the seller under \MPV,
  and $R$ the seller's revenue.
By optimality of $i^*$ for the seller, we get that

\begin{align}\label{eq:rev-q}
R & = \Val{i^*} \cdot (\MP{i^*} + \sum_{k=i^*+1}^n \MP{k}) 
\; \geq \; \Val{i^*+1} \cdot \sum_{k=i^*+1}^n \MP{k}
\; = \; 2 \Val{i^*} \cdot \sum_{k=i^*+1}^n \MP{k}.
\end{align}

Hence, $\MP{i^*} \geq \sum_{k=i^*+1}^n \MP{k}$,
meaning that at least half of the probability mass, and thus half of
the seller's revenue $R$ from \MPV, comes from type $i^*$.
Now define a segment \MPV['] of the original type distribution \PVec['] as follows. 
For type $i^*$, take a total of probability mass $\MP{i^*}$ from
types $j$ with $2^{i^*} \leq \Val[']{j} < 2^{i^*+1}$.
For all types $i > i^*$, 
take a total of probability mass $\MP{i}/8$ from
types $j$ with $2^{i} \leq \Val[']{j} < 2^{i+1}$,
and for types $j$ with $\Val[']{j}< 2^{i^*}$, set $\MP[']{j} = 0$.
Observe that the \FISW of \MPV['] is at least a $1/8$ fraction of the \ServWel of \MPV.
It remains to show that a constant fraction of that social welfare is
above the seller's revenue-maximizing offer price for \MPV['], which
we do next.

First, note that the price $2^{i^*}$ gives the seller a revenue of at
least $R/2$, even just from all buyers of types $j$ with 
$2^{i^*} \leq \Val[']{j} < 2^{i^*+1}$.
Next, consider a price of $\Val[']{j} \geq 2^{i^*+1}$, 
say $2^i \leq \Val[']{j} \leq 2^{i+1}$ with $i \geq i^* + 1$.
The revenue of such a price is at most
\begin{align}\label{eq:highvals-rev}
\Val[']{j} \cdot \sum_{k \geq j} \MP[']{k} 
& \leq  2^{i+1} \cdot \sum_{k \geq i} \MP{k}/8
\; = \; \quarter v_i \cdot \sum_{k \geq i} \MP{k} 
\; \leq \; R/4.
\end{align} 
Hence, no price $\Val[']{j} \geq 2^{i^* + 1}$ can be
revenue-maximizing for \MPV['], and all types in the segment \MPV[']
with value at least $2^{i^*+1}$ are served. 
It remains to show that a constant factor of the remaining social
welfare --- associated with values between $2^{i^*}$ and $2^{i^*+1}$
--- is also served.

Because the price $2^{i^*}$ dominates all prices 
$\Val[']{j} \geq 2^{i^* + 1}$,
the seller will choose some price $j$ with 
$2^{i^*} \leq \Val[']{j} < 2^{i^* + 1}$,
and the chosen price must give the seller revenue at least $R/2$.
The calculation in \eqref{eq:highvals-rev} implies that the revenue
extracted from types $j$ with $\Val[']{j} \geq 2^{i^* + 1}$ is at most $R/4$.
Hence, at least a revenue (and thus also social welfare) of 
$R/4$ must come from types $j$ with $2^{i^*} \leq \Val[']{j} < 2^{i^* + 1}$. 
By construction, the \FullInfWel associated with
those types is at most $2^{i^*+1} \MP{i^*} = 2\Val{i^*} \MP{i^*}$, 
which is at most $2R$ by \eqref{eq:rev-q}. 
Consequently, at least a $1/8$ fraction of the total social welfare
associated with those types is served, as needed.

In summary, \MPV['] serves a constant fraction of the \FISW
of \MPV, which in turn is a $1/O(\log n) \geq 1/O(\log m)$
fraction of the \FISW of the distribution \PVec[']. 
Hence, for any distribution supported on $m$ types, a single segment is enough to serve a
$1/O(\log m)$ fraction of the social welfare. 

\medskip

To show that this bound is tight, we construct an
``equal-welfare distribution'' with the property that no single
segment extracts more than an $O(1/\log \EqualRevN)$ fraction of
the full-information social welfare. 
Let $\Val{i} = 2^i$ for $0 \leq i \leq \EqualRevN$, 
and $\PVal{i} = \frac{1}{2^i \cdot (\EqualRevN-i)}$ for $0 \leq i < \EqualRevN$, 
with $\PVal{\EqualRevN} = \frac{1}{2^{\EqualRevN-1}} = \PVal{\EqualRevN-1}$.
Notice that these ``probabilities'' do not sum to 1;
we omit the normalization constants for legibility, since they will
cancel out in the subsequent calculations.

The \FISW of \PVec is 
$\sum_{i=0}^\EqualRevN \PVal{i}\Val{i} = 2 + \sum_{i=1}^\EqualRevN \frac{1}{i} = \Theta(\log \EqualRevN)$. 
Now consider any segment \MPV of \PVec, and let $\Val{k}=2^k$
be the seller's revenue-maximizing price for \MPV. 
If $k=\EqualRevN$, the \ServWel is 
$\Val{\EqualRevN} \MP{\EqualRevN} \leq \Val{\EqualRevN}\PVal{\EqualRevN} = 2=O(1)$, as needed. 
Assume now that $k < \EqualRevN$. 
For each $i$ with $k < i \leq \EqualRevN$, the optimality of the price \Val{k} implies that 
$\Val{i} \MP{[i,\EqualRevN]} \leq \Val{k} \MP{[k,\EqualRevN]}$, 
or equivalently $\MP{[i,\EqualRevN]} \leq \frac{\Val{k}}{\Val{i}} \MP{[k,\EqualRevN]}$. 
By choosing $i=k+1$, we get that 
$\MP{[k+1,\EqualRevN]} \leq \frac{\Val{k}}{\Val{k+1}} \MP{[k,\EqualRevN]} = \frac{\MP{[k,\EqualRevN]}}{2}$, 
and consequently $\MP{[k,\EqualRevN]} \leq 2\MP{k}$. 
We can now bound the \ServWel of \MPV by a constant:

\begin{align*}
\sum_{i=k}^\EqualRevN \Val{i} \MP{i} 
& = \Val{k} \MP{[k,\EqualRevN]} + \sum_{i=k+1}^\EqualRevN (\Val{i} - \Val{i-1}) \MP{[i,\EqualRevN]}
\; \leq \; 2 \Val{k} \MP{k} + 2 \sum_{i=k+1}^\EqualRevN (\Val{i} - \Val{i-1})  \frac{\Val{k}}{\Val{i}} \MP{k} \\
& = 2 \Val{k} \MP{k} \left(1+ \sum_{i=k+1}^\EqualRevN \frac{\Val{i} - \Val{i-1}}{\Val{i}} \right) 
\; = \; \Val{k} \MP{k} (\EqualRevN+2-k).
\end{align*}
Because $\MP{k} \leq \PVal{k}$, we can bound
$\Val{k} \MP{k} (\EqualRevN+2-k) \leq \Val{k} \PVal{k} (\EqualRevN+2-k)
= \frac{\EqualRevN+2-k}{\EqualRevN-k} \leq 3$.
In summary, any single segment can obtain at most constant welfare,
and thus at most an $O(1/\log \EqualRevN)$ fraction of the total social
welfare for this instance.
\end{extraproof}

To derive Corollary~\ref{cor:logn-signal-welfare}, 
simply pick segments greedily, always choosing the next
segment to have largest possible \ServWel.
By Theorem~\ref{thm:single-signal-welfare}, each subsequent segment
obtains at least a $\frac{1}{c\log n}$ fraction of the residual
welfare at that point, meaning that after adding
$c \log n \log (1/\eps)$ signals, the fraction of the total
welfare obtained by the signaling scheme is at least
$1-(1-\frac{1}{c\log n})^{c\log n\log (1/\epsilon)} \geq 1-\epsilon$.

\begin{remark}
One can obtain a positive result very directly when the ratio
$\rho = \frac{\max_i \Val{i}}{\min_i \Val{i}}$ is bounded.
Group all buyers according to their values into $\log \rho$ bins: 
for $u=\min_i \Val{i}$, buyers with $2^{j}u\le\Val{i}<2^{j+1}u$ are
put into bin $j$, 
By considering each bin as a segment, we can see that the revenue
(also, the \ServWel) of each segment is at least half of the its \TotWel. 
Choosing the best \NUMSIG bins leads to a
bound of $\Omega(\min(1,\NUMSIG / \log \rho))$.
However, our result is of more interest when the ratio between the
largest and smallest valuations can be very large.
\end{remark}



To get a quasi-polynomial time approximation scheme (QPTAS) for the
problem of finding a welfare-maximizing signaling scheme with at most
\NUMSIG signals, consider  a desired approximation parameter
$\eps$.
We distinguish two cases.
If $\NUMSIG \leq c \log n \log (1/\eps)$ 
(again, $c$ is the constant in $O(\log n)$ in Theorem~\ref{thm:single-signal-welfare}), 
enumerate all possible choices of price points, and find the truly
optimal one (i.e., there is no approximation factor lost in this case).
Notice that there are at most ${n \choose \NUMSIG} = O(n^\NUMSIG)$
such combinations to consider, and for each of them, 
either the LP~\eqref{lp:welfare-maximization} or the greedy algorithm
(Theorem~\ref{thm:greedy-optimal}) allows us to evaluate the maximum
welfare attainable with that set of prices.
Thus, the running time is quasi-polynomial for fixed $\epsilon$.

When $\NUMSIG > c \log n \log (1/\eps)$, 
by Corollary~\ref{cor:logn-signal-welfare},
choosing  the \NUMSIG price points greedily obtains a $(1-\eps)$
fraction of the \FISW.
Thus, it certainly obtains the same fraction of the maximum 
that could be achieved with \NUMSIG signals.

\section{\PWF Maximization with a Greedy Algorithm}

\label{sec:greedy}
In this section, we show that, given the set of price points,
an optimal solution to the LP~\eqref{lp:welfare-maximization} 
(i.e., the problem of maximizing \pwf) can be computed by a greedy
algorithm constructing the signals one by one.
Besides the faster running time, the main value of the greedy
algorithm is as an analysis tool for the problem of choosing the 
(near-)optimal price points; the analysis for that problem is carried
out in Section~\ref{sec:submodularity}.

We begin with an algorithm for constructing just one signal \SIG with a
given price point $\PrP = \SigI{\SIG}$.
Let \RPVec be a vector of residual probabilities for buyer types,
i.e., the probability of each type that has not been allocated to any
signals previously. Hence, $\vc{0} \leq \RPVec \leq \PVec$.
The goal is to construct the probabilities \OSVec for a single signal
with price point \PrP; that is, we require that
$0 \leq \OS{i} \leq \RP{i}$ for all $i$, and the revenue constraint 
$\Val{\PrP} \cdot \sum_{j \geq \PrP} \OS{j} 
\geq \Val{i} \cdot \sum_{j \geq i} \OS{j}$ must be satisfied for all $i$.
The revenue constraint can be written 
as
\begin{align}
\TS{i} & \leq 
\frac{\Val{\PrP}}{\Val{i}} \cdot \TS{\PrP} \quad \mbox{ for all } i.
\label{eqn:tailsum-inequality}
\end{align}

In an optimal (single) signal, for each $i \geq k$, either
Inequality~\eqref{eqn:tailsum-inequality} must be tight, or $\OS{i} = \RP{i}$;
otherwise, \OS{i} could be increased, raising social welfare.
For the same reason, 
$\OS{\PrP} = \RP{\PrP}$.
Because no buyer of type $i < \PrP$ will ever purchase at price
\Val{\PrP}, 
such buyers contribute 0 to revenue and (social or buyer) welfare. 
Hence, without loss of generality, 
the optimum signal has $\OS{i} = 0$ for all $i < \PrP$.

These observations suggest the following algorithm: 
Gradually raise \OS{\PrP} from 0 to \RP{\PrP}. 
Simultaneously raise the \OS{i} for $i > \PrP$ in such
a way that at all times, for each $i$, either $\OS{i} = \RP{i}$ or
Inequality~\eqref{eqn:tailsum-inequality} is tight.
This is accomplished by raising the tail sums \TS{i} at a rate of
$\frac{\Val{\PrP}}{\Val{i}}$ while \TS{\PrP} is raised at rate 1.
Solving for the necessary rates of increase of individual \OS{i} gives
rise to the algorithm \algoname{Construct-One-Signal}.

\begin{algorithm}
\caption{Construct-One-Signal $(\PrP, \VVec, \RPVec)$ \label{algo:one-signal}}
\begin{algorithmic}
    \STATE $\OSVec \gets \vc{0}$
    \STATE $I \gets \Set{i \geq \PrP}{\RP{i} > 0}$ 
         \COMMENT{Throughout, $I$ contains all types for which \OS{i} can still be raised.}
    \WHILE{$\PrP \in I$} 
      \STATE $m \gets \max_{i \in I} i$
      \STATE $\Rate{m} \gets \frac{1}{v_m}$ 
         \COMMENT{Rate of increase that keeps Inequality~\eqref{eqn:tailsum-inequality} tight for $m$.}
      \FORALL {$i \in I \setminus \SET{m}$}
        \STATE $j \gets \min \Set{i' > i}{i' \in I}$
        \STATE $\Rate{i} \gets \frac{1}{\Val{i}}-\frac{1}{\Val{j}}$
         \COMMENT{Rate of increase that keeps Inequality~\eqref{eqn:tailsum-inequality} tight for $i$.}
      \ENDFOR
      \STATE $i^{*} \gets \arg \min_{i \in I}(\RP{i} - \OS{i})/\Rate{i}$
         \COMMENT{Index for which $\RP{i^*} = \OS{i^*}$ will become tight first.}
      \STATE $\alpha \gets (\RP{i^*} - \OS{i^*})/\Rate{i^*}$
         \COMMENT{Amount of increase in \OS{\PrP} until $\OS{i^*} = \RP{i^*}$.}
      \FORALL{$i \in I$}
         \STATE $\OS{i} \gets \OS{i} + \alpha \Rate{i}$
      \ENDFOR
      \STATE $I \gets \Set{i \geq \PrP}{\OS{i} < \RP{i}}$ 
         \COMMENT{Update the set of indices that can still be raised.}
    \ENDWHILE
    \RETURN \OSVec
\end{algorithmic}
\end{algorithm}

Our first lemma simply restates the revenue constraint, 
and captures the fact that if $\OS{i} < \RP{i}$ at any time
during the algorithm, then the revenue constraint
\eqref{eqn:tailsum-inequality} must be tight for $i$.
The proof is straightforward by induction on iterations.

\begin{lemma} \label{lem:revenue-tight}
For all indices $i$ and every step of the algorithm, 
$\TS{i} \leq \frac{\Val{\PrP} \cdot \PMY{[\PrP,i)}}{\Val{i} - \Val{\PrP}}$. 

Furthermore, if $i \in I$ at some point of the algorithm (including at
termination), then at that point in time, 
$\Val{i} \cdot \TS{i} = \Val{\PrP} \cdot \TS{\PrP}$.
The latter condition is equivalent to saying that
$\TS{i} = \frac{\Val{\PrP} \PMY{[\PrP,i)}}{\Val{i} - \Val{\PrP}}$.
\end{lemma}


The key property of the algorithm \algoname{Construct-One-Signal} is
that it maximizes all tail sums \TS{i}:

\begin{lemma} \label{lem:maxtail}
Let $\OSVec['] \leq \RPVec$ be any signal.
If \PrP is the chosen price by the seller under \OSVec['], i.e.,
$\Val{\PrP} \cdot \TS[']{\PrP} \geq \Val{i} \cdot \TS[']{i}$ for all $i$,
then $\TS{i} \geq \TS[']{i}$ for all $i \geq \PrP$.
\end{lemma}

\begin{emptyproof}
For contradiction, assume that $\TS{i} < \TS[']{i}$ for some $i$, and
fix the minimum such $i$.
Let $j \geq i$ be the minimum index with $\OS{j} < \RP{j}$; note that
such an index must exist, as otherwise, 
$\TS{i} = \sum_{j \geq i} \OS{j} = \sum_{j \geq i} \RP{j} 
\geq \sum_{j \geq i} \OS[']{j} = \TS[']{i}$.
As a result, the index $j \in I$ at the termination of the algorithm,
implying by Lemma~\ref{lem:revenue-tight} that 
$\Val{\PrP} \cdot \TS{\PrP} = \Val{j} \cdot \TS{j}$.
Now, we obtain the following contradiction:
\begin{align*}
\TS{i} & = \TS{j} + \sum_{i'=i}^{j-1} \OS{i'}
\; = \; \frac{\Val{\PrP} \cdot \TS{\PrP}}{\Val{j}} + \sum_{i'=i}^{j-1} \RP{i'}
\; \geq \; \TS[']{j} + \sum_{i'=i}^{j-1} \OS[']{i'}
\; = \; \TS[']{i}. \hfill \QED
\end{align*}
\end{emptyproof}

To construct a complete signaling scheme, we invoke 
the algorithm \algoname{Construct-One-Signal} repeatedly, constructing the
signals one at a time. The important part here is that the signals
must be constructed in decreasing order of the target price.
The intuitive reason is that the inclusion of high-value buyers in a
signal makes a higher price more attractive to the seller, thus posing
additional constraints on the required probability mass of
lower-valued buyers that must be included. Thus, it is always better
to include as many high-valued buyers in the high-priced signals as
possible, and this is accomplished by constructing those signals first.
Hence, we assume that the signals are sorted in descending order of
their prices.

\begin{algorithm}[htb]
\caption{Construct-Signaling-Scheme 
  $(\VVec, \PVec, S = \SET{\SigI{1} > \SigI{2} > \ldots > \SigI{\NUMSIG}})$}
\begin{algorithmic}
 \STATE $\RPVec[1] \gets \PVec$
 \FOR {$\SIG \gets 1$ to \NUMSIG}
    \STATE $\SigC{\SIG} \gets \algoname{Construct-One-Signal}(\RPVec[\SIG], \VVec, \SigI{\SIG})$
        \COMMENT{\SigC{\SIG} is column \SIG of the signaling scheme \SigScheme's matrix.}
    \STATE $\RPVec[\SIG+1] \gets \RPVec[\SIG] - \SigC{\SIG}$
 \ENDFOR
 \RETURN \SigScheme
\end{algorithmic}
\end{algorithm}

\begin{theorem} \label{thm:greedy-optimal}
The algorithm \algoname{Construct-Signaling-Scheme} solves 
the linear program \eqref{lp:welfare-maximization} optimally.
\end{theorem}

The proof of this theorem proceeds by showing that an optimal solution
\SigScheme[*] for the linear program \eqref{lp:welfare-maximization}
can be (gradually) transformed into the solution \SigScheme
constructed by the algorithm \algoname{Construct-Signaling-Scheme}
without decreasing its solution quality.
It is technically fairly involved, and given in
Appendix~\ref{sec:greedy-proof}.

\section{Submodularity of \PWF}

\label{sec:submodularity}

We prove that the \pwf objective function
\PWelfare{S} is a submodular function of the set $S$ of chosen price
points. (There is also a garbage signal $\GSIG \notin S$.)
Let $S = \SET{\SigI{1} > \SigI{2} > \ldots > \SigI{m}}$
be $k$ price points. 
\PWelfare{S} is defined as the optimum
solution to the LP~\eqref{lp:welfare-maximization} with the given
price points. We show the following:

\begin{theorem} \label{thm:submodularity}
If $S \subseteq T$ and $k \notin T$, then 
$\PWelfare{T \cup \SET{k}} - \PWelfare{T} 
\leq \PWelfare{S \cup \SET{k}} - \PWelfare{S}$.
\end{theorem}

Theorem~\ref{thm:greedy-optimal} states that the greedy algorithm
solves the problem for $S \cup \SET{k}$ optimally, but the effects of
adding $k$ are subtle.
It is fairly easy to analyze what happens in the iteration when $k$
itself is added: that the welfare increase is larger for $S$ than for
$T$ is easily seen by a simple monotonicity argument, captured by
Lemma~\ref{lem:pointwise-dominance}.
However, the addition of $k$ has ``downstream'' effects. 
The construction of subsequent signals with price points 
$\SigI{\SIG} < k$ will now face different residual probabilities, and
the resulting reductions in those signals need to be carefully
balanced against the gains from the signal with price point $k$.

Part of the complexity arises from the rather complex construction
of the signal for price point $k$ itself. It is captured by the
algorithm \algoname{Construct-One-Signal}, which itself runs through
iterations in which different sets $I$ of indices have their
probabilities increased. In order to eliminate this source of
complexity, we will think of adding the signal with price point $k$
``gradually.'' Specifically, we consider the execution of 
\algoname{Construct-Signaling-Scheme} in which the execution of
\algoname{Construct-One-Signal} for the signal with price point $k$
may be terminated prematurely. An upper bound $B$ on the tail
probability is specified, and \algoname{Construct-One-Signal} is
stopped when $\sum_{i \geq k} \OS{i} = B$.
After signal $k$ is constructed in this modified way, subsequent
signals will be constructed normally by
\algoname{Construct-Signaling-Scheme}. 

A modification of the proof of Theorem~\ref{thm:greedy-optimal} shows
that this modified algorithm optimally solves the 
LP~\eqref{lp:welfare-maximization} with the added constraint that
$\sum_{i \geq k} \SigP{i}{\SIG[k]} \leq B$, where \SIG[k] denotes the
signal whose price point is $k$.

We write \PWelfare[(k,B)]{S} for the
\pwf achieved by the optimum solution with a set of
signal price points $S \cup \SET{k}$, and the constraint that the 
probability mass for signal $k$ is at most $B$.
Our main lemma is: 

\begin{lemma} \label{lem:small-increase-welfare}
If $S \subseteq T$, then for any $k, B, \epsilon$: 
$\PWelfare[(k,B+\epsilon)]{T} - \PWelfare[(k,B)]{T}
\leq \PWelfare[(k,B+\epsilon)]{S} - \PWelfare[(k,B)]{S}$.
\end{lemma}

Lemma~\ref{lem:small-increase-welfare} implies submodularity quite
directly, as follows.

\begin{extraproof}{Theorem~\ref{thm:submodularity}}
Let \PMS and \PMT be the optimal signaling schemes with
price point sets $S \cup \SET{k}$ and $T \cup \SET{k}$, respectively. 
By Lemma~\ref{lem:pointwise-dominance}, when constructing $\SIGS[k]$ and
$\SIGT[k]$, the residual probability for \SIGS[k] is more than that
of \SIGT[k]; therefore, 
$\pmT{\SIGT[k]}{[k,\NUMVAL]} \leq \pmS{\SIGS[k]}{[k,\NUMVAL]}$ 
by Lemma~\ref{lem:maxtail}. 

Consider gradually increasing $B$ from $0$ to 
\pmT{\SIGT[k]}{[k,\NUMVAL]} in increments of (varying) $\epsilon$, as
outlined above. Subsequently, continue increasing $B$ for \PMS only.
By adding up the inequality from Lemma~\ref{lem:small-increase-welfare}
for each such step, and noting that the subsequent increases of $B$
for \PMS can only further increase the welfare of \PMS,
we obtain that 
$\PWelfare{S\cup\SET{k}} - \PWelfare{S} \geq
\PWelfare{T\cup\SET{k}} - \PWelfare{T}$. 
\end{extraproof}

Because the objective function is submodular (and monotone by
Lemma~\ref{lem:pointwise-dominance}), the greedy algorithm is known
\cite{nemhauser:wolsey:fisher} to give a $(1-1/e)$-approximation for
the problem of maximizing \PWelfare{S}. 
Proposition~\ref{prop:welfare-sanitized} then implies
that the same greedy algorithm gives a 
$\frac{\NUMSIG-1}{\NUMSIG} \cdot (1-1/e)$ approximation for the
objective of maximizing the social welfare \Welfare{S}, proving
Theorem~\ref{thm:welfare-intro}.

\smallskip

The proof of Lemma \ref{lem:small-increase-welfare} is
technically quite involved, and given in 
Appendix~\ref{sec:submodularity-proofs}.
The idea is to first prove it for sufficiently small
$\epsilon$, which allows us to couple the executions tightly. 

In particular, by comparing the solutions to the linear
program~\eqref{lp:welfare-maximization}, we can ensure that any
constraint that becomes tight in the solution for set $T$ with
bound $B + \epsilon$, but is not tight with bound $B$ (and similarly
for $S$) would not have become tight for any $\epsilon' < \epsilon$.
This will localize the changes, and maintain the revenue indifference for the
seller.
By summing over all such iterations (there will only be finitely
many, because $\epsilon$ is chosen so that at least one more
constraint becomes tight), we eventually prove the lemma for all
$\epsilon$.

In the analysis, we are interested in four different signaling
schemes, constructed by \algoname{Construct-Signaling-Scheme} when run
with different sets of price points and upper bounds $B$.
We will assume here that $k \in S \subseteq T$.
Specifically we define:

\begin{itemize}
\item $\PMS = (\pmS{\SIG}{i})_{\SIG,i}$: the probability mass of
  type $i$ assigned to signal \SIG when the algorithm is run with
  price point set $S$ and an upper bound of $B$.
\item $\PMSP = (\pmSp{\SIG}{i})_{\SIG,i}$: probability mass for price
  point set $S$ and an upper bound of $B+\epsilon$.
\item $\PMT = (\pmT{\SIG}{i})_{\SIG,i}$: probability mass for price
  point set $T$ and an upper bound of $B$.
\item $\PMTP = (\pmTp{\SIG}{i})_{\SIG,i}$: probability mass for price
  point set $T$ and an upper bound of $B+\epsilon$.
\end{itemize}

A first step of the proof, which is both illustrative of the types of
arguments made repeatedly and implies monotonicity of the
\pwf \PWelfare{S} is the following lemma.
It shows that if more probability mass can be allocated for one
signal, then the effect can never be a decrease in the total allocated
probability mass for any type of buyer.

\begin{lemma}[Monotonicity] \label{lem:pointwise-dominance}
Let \PMS, \PMSP be optimal signaling schemes for price point set $S$,
with signal set \SIGSETS.
Then, for any $k \in S, B, \epsilon$, and any buyer type $i$, 
we have that $\opmS{\SIGSETS}{i} \leq \opmSp{\SIGSETS}{i}$.

As a corollary, for any sets of price points $S \subseteq T$,
and with \SIGSETS and \SIGSETT as the respective sets of signals, 
for any $B$ and any buyer type $i$, we have 
$\opmS{\SIGSETS}{i} \leq \opmT{\SIGSETT}{i}$.
\end{lemma}

\begin{emptyproof}
Let \SIG[k] be the signal with price point $k$.
Let $\SIGSET['] \subseteq \SIGSET$ be an initial segment of \SIGSET, i.e., 
$\SIG['] < \SIG$ for every $\SIG['] \in \SIGSET['], \SIG \in \SIGSET \setminus \SIGSET[']$.
We prove by induction on $\SetCard{\SIGSET[']}$ that
$\opmS{\SIGSET[']}{i} \leq \opmSp{\SIGSET[']}{i}$.
The base case $\SIGSET['] = \emptyset$ is trivial.

For the induction step, let \SIG['] be the largest signal in $\SIGSET[']$
(i.e., with smallest price point $k'$), and
distinguish three cases.
If $\SIG['] < \SIG[k]$, i.e., before \SIG[k] is constructed, the
execution of \algoname{Construct-Signaling-Scheme} 
is the same for the bounds $B$ and $B+\epsilon$, so 
$\pmS{\SIG[']}{i} = \pmSp{\SIG[']}{i}$, and the induction follows directly.

When $\SIG['] = \SIG[k]$, $\pmS{\SIG[k]}{i}$ is the result of
\algoname{Construct-One-Signal} with an upper bound of $B$, and
$\pmSp{\SIG[k]}{i}$ is that with upper bound of $B+\epsilon$. 
Until the tail sum reaches $B$, the execution is the same, and
subsequently, values can only be raised further for the execution with
the bound $B+\epsilon$.
Thus, $\pmS{\SIG[k]}{i} \le \pmSp{\SIG[k]}{i}$, and
again, the induction step follows.

For $\SIG['] > \SIG[k]$, assume for contradiction that there exists an
$i$ such that $\opmS{\SIGSET[']}{i} > \opmSp{\SIGSET[']}{i}$; fix the smallest such $i$.
Since $\PVal{i} \geq \opmS{\SIGSET[']}{i} > \opmSp{\SIGSET[']}{i}$, we get that the
index $i$ is in $I$ for the execution with upper bound $B+\epsilon$
all the way until \SIG['] is constructed, implying by
Lemma~\ref{lem:revenue-tight} (and the revenue constraint for the
upper bound of $B$) that for all $\SIG \leq \SIG[']$,
\begin{align*}
\pmS{\SIG}{[i,\NUMVAL]} & \leq
\frac{\Val{\SigI{\SIG}}}{\Val{i} - \Val{\SigI{\SIG}}} 
    \cdot \pmS{\SIG}{[\SigI{\SIG}, i)},
&
\pmSp{\SIG}{[i,\NUMVAL]} & =
\frac{\Val{\SigI{\SIG}}}{\Val{i} - \Val{\SigI{\SIG}}} 
    \cdot \pmSp{\SIG}{[\SigI{\SIG}, i)}.
\end{align*}

By summing over all $\SIG \in \SIGSET[']$, we also obtain that
\begin{align*}
\opmS{\SIGSET[']}{[i,\NUMVAL]} & \leq
\sum_{\SIG \in \SIGSET[']} \frac{\Val{\SigI{\SIG}}}{\Val{i} - \Val{\SigI{\SIG}}} 
    \cdot \pmS{\SIG}{[\SigI{\SIG}, i)},
&
\opmSp{\SIGSET[']}{[i,\NUMVAL]} & =
\sum_{\SIG \in \SIGSET[']} \frac{\Val{\SigI{\SIG}}}{\Val{i} - \Val{\SigI{\SIG}}} 
    \cdot \pmSp{\SIG}{[\SigI{\SIG}, i)}.
\end{align*}

By minimality of $i$, we have that $\opmS{\SIGSET[']}{i'} \leq \opmSp{\SIGSET[']}{i'}$
for all $i' < i$; summing these inequalities gives that
$\opmS{\SIGSET[']}{[\SigI{\SIG}, i)} \leq \opmSp{\SIGSET[']}{[\SigI{\SIG}, i)}$.
Similarly, for any initial segment $\SIGSET[''] \subsetneq \SIGSET[']$, the strong
induction hypothesis implies that $\opmS{\SIGSET['']}{i'} \leq \opmSp{\SIGSET['']}{i'}$
for all $i'$ (in particular, $i' < i$); summing those inequalities,
and combining with the one just derived proves that 
$\opmS{\SIGSET['']}{[\SigI{\SIG}, i)} \leq \opmSp{\SIGSET['']}{[\SigI{\SIG}, i)}$
for all initial segments $\SIGSET[''] \subseteq \SIGSET[']$ (including $\SIGSET['']=\SIGSET[']$).

Because $\frac{\Val{\SigI{\SIG}}}{\Val{i} - \Val{\SigI{\SIG}}}$ is
monotone non-increasing in \SIG (recall that \SigI{\SIG} is
decreasing), we can apply Lemma~\ref{lem:prefix-sums-imply-combination}
in the middle step of the following derivation:

\begin{align} \label{eq:lessthanplus}
\opmSp{\SIGSET[']}{[i,\NUMVAL]} & =
\sum_{\SIG \in \SIGSET[']} \frac{\Val{\SigI{\SIG}}}{\Val{i} - \Val{\SigI{\SIG}}} 
    \cdot \pmSp{\SIG}{[\SigI{\SIG}, i)}.
\; \stackrel{\text{Lemma~\ref{lem:prefix-sums-imply-combination}}}{\geq} \;
\sum_{\SIG \in \SIGSET[']} \frac{\Val{\SigI{\SIG}}}{\Val{i} - \Val{\SigI{\SIG}}} 
    \cdot \pmS{\SIG}{[\SigI{\SIG}, i)},
\; \geq \; \opmS{\SIGSET[']}{[i,\NUMVAL]}.
\end{align}

Because $\opmS{\SIGSET[']}{i} > \opmSp{\SIGSET[']}{i}$, but
$\opmS{\SIGSET[']}{[i,\NUMVAL]} \leq \opmSp{\SIGSET[']}{[i,\NUMVAL]}$, there must be some
$j > i$ such that $\opmS{\SIGSET[']}{j} < \opmSp{\SIGSET[']}{j}$; fix a minimal such $j$.
This time, since $\PVal{j} \geq \opmSp{\SIGSET[']}{j} > \opmS{\SIGSET[']}{j}$,
we can apply Lemma~\ref{lem:revenue-tight} (and the revenue constraint
for the process with bound $B+\epsilon$) to obtain that for all $\SIG
\leq \SIG[']$,

\begin{align*}
\Val{j} \cdot \pmS{\SIG}{[j,\NUMVAL]} 
& = \Val{\SigI{\SIG}} \cdot \pmS{\SIG}{[\SigI{\SIG}, \NUMVAL]} 
\; \ge \; \Val{i} \cdot \pmS{\SIG}{[i,\NUMVAL]}, 
&
\Val{j} \cdot \pmSp{\SIG}{[j,\NUMVAL]} 
& \le \Val{\SigI{\SIG}} \cdot \pmSp{\SIG}{[\SigI{\SIG}, \NUMVAL]} 
\; = \; \Val{i} \cdot \pmSp{\SIG}{[i,\NUMVAL]}.
\end{align*}

Solving for \pmS{\SIG}{[j,\NUMVAL]} and \pmSp{\SIG}{[j,\NUMVAL]}, 
we obtain that
$\pmS{\SIG}{[j,\NUMVAL]} \geq \frac{\Val{j}}{\Val{j} - \Val{i}} \cdot \pmS{\SIG}{[i,j)}$ 
and $\pmSp{\SIG}{[j,\NUMVAL]} \leq \frac{\Val{j}}{\Val{j} - \Val{i}} \cdot \pmSp{\SIG}{[i,j-1]}$. 
Summing over all $\SIG \in \SIGSET[']$ now gives us that
\begin{align} \label{eq:greaterthanplus}
\opmS{\SIGSET[']}{[j,\NUMVAL]} 
& \geq \frac{\Val{j}}{\Val{j} - \Val{i}} \cdot \opmS{\SIGSET[']}{[i,j)},
&
\opmSp{\SIG[']}{[j,\NUMVAL]} 
& \leq \frac{\Val{j}}{\Val{j} - \Val{i}} \cdot \opmSp{\SIGSET[']}{[i,j)}.
\end{align}

By the definition of $i$ and $j$, we have that
$\opmS{\SIGSET[']}{[i,j)} > \opmSp{\SIGSET[']}{[i,j)}$; substituting this inequality
into \eqref{eq:greaterthanplus} and canceling common terms implies that
$\opmS{\SIGSET[']}{[j,\NUMVAL]} > \opmSp{\SIGSET[']}{[j,\NUMVAL]}$.
Now, we derive a contradiction as follows:
by Inequality~\eqref{eq:lessthanplus}, we have 
\begin{align*}
\opmS{\SIGSET[']}{[j,\NUMVAL]} 
& = \opmS{\SIGSET[']}{[i,\NUMVAL]} - \opmS{\SIGSET[']}{[i,j)}
\; \stackrel{\eqref{eq:lessthanplus}}{<} \; \opmSp{\SIGSET[']}{[i,\NUMVAL]} - \opmSp{\SIGSET[']}{[i,j)}
\; = \; \opmSp{\SIGSET[']}{[j,\NUMVAL]}
\; < \; \opmS{\SIGSET[']}{[j,\NUMVAL]}. \hfill \QED
\end{align*}
\end{emptyproof}

\section{Revenue in Bilateral trade}
\label{sec:seller-revenue}

In this section, we prove Theorem~\ref{thm:revenue-intro},
giving a straightforward dynamic program to compute
a signaling scheme maximizing the seller's revenue. 
Before doing so, we exhibit an equal-revenue
distribution for which any signaling scheme with \NUMSIG
signals only recovers an $O(\NUMSIG/n)$ fraction of the \FISW.
Notice again the contrast to the case of welfare maximization,
where even one segment is enough to attain an $\Omega(1/\log n)$
fraction of the \FISW.

We define an equal-revenue distribution as follows.
Let the valuations be $\Val{i} = 2^i$ for $0 \leq i \leq n$, 
and the probabilities 
$\PVal{i} = \frac{1}{2^{i+1}}$ for $0 \leq i < n$, 
and $\PVal{n}= \frac{1}{2^n}$. 
The \FISW is $\sum_{i=0}^n \PVal{i}\Val{i} = \frac{n}{2} + 1$. 
However, for every segment $\MPV \leq \PVec$, no matter what
price \Val{i} the seller chooses, the 
revenue cannot be more than $\Val{i} \cdot \sum_{j \geq i} \PVal{j} = 1$.
Therefore, in the worst case, with \NUMSIG signals, the seller can
at best get an $O(\NUMSIG/n)$ fraction of the maximum social welfare
as his revenue, whereas a fully informed seller would be able
to extract the entire \FISW, as discussed in the introduction.

\medskip

Next, we turn our attention to the problem of computing the optimum
signaling scheme for the seller's revenue.
The key insight enabling a dynamic program is that the
seller-optimal signaling scheme partitions the buyer types into disjoint
intervals, and allocates all probability mass for a given interval to
one signal.

\begin{lemma}[Interval Structure of Seller-Optimal Signaling Scheme] 
\label{lem:intervals}
W.l.o.g., the seller-optimal signaling scheme \SigScheme has the following form: 
There are disjoint intervals $\Intv{1}, \Intv{2}, \ldots, \Intv{\NUMSIG}$
of buyer types 
such that $\bigcup_{\SIG} \Intv{\SIG} = \SET{1, \ldots, \NUMVAL}$,
and for each signal \SIG, $\SigP{i}{\SIG} = \PVal{i}$ for all
$i \in \Intv{\SIG}$ (and $\SigP{i}{\SIG} = 0$ for all $i \notin \Intv{\SIG}$).
\end{lemma}

\begin{proof}
Let $\SigI{1} > \SigI{2} > \ldots > \SigI{\NUMSIG}$ be the price
points of the signals \SIG under \SigScheme. We will show how to
transform \SigScheme to the claimed form without decreasing the
seller's revenue.

First, if $\SigP{i}{\SIG} > 0$ for some $\SIG < \NUMSIG, i < \SigI{\SIG}$,
then the buyers of type $i$ will not buy when signal \SIG is sent,
contributing nothing to the seller's revenue. Therefore, setting
$\SigP{i}{\SIG} = 0$ instead does not lower the seller's revenue, and
increasing \SigP{i}{\NUMSIG} by the same amount again cannot decrease the
seller's revenue. Hence, we may assume that for all signals
$\SIG < \NUMSIG$, we have $\SigP{i}{\SIG} > 0$ only for $i \geq \SigI{\SIG}$.

Next, if $\SigP{i}{\SIG} > 0$, then $\SigP{i}{\SIG} = \PVal{i}$.
We distinguish two cases: if there is unallocated probability mass of
type $i$, then \SigP{i}{\SIG} can simply be raised.
If $\SigP{i}{\SIG[']} > 0$ for $\SIG['] > \SIG$, we can lower
\SigP{i}{\SIG[']} to 0 while raising \SigP{i}{\SIG} by the same
amount. Because $\SIG < \SIG['] \le \NUMSIG$,  
we have that $i \geq \SigI{\SIG}$, so the seller's revenue
increases by 
$\SigP{i}{\SIG[']} \cdot (\Val{\SigI{\SIG}} - \Val{\SigI{\SIG[']}}) \geq 0$.

So far, we have shown that the signals partition the buyer types into
sets such that for each buyer type, all of its probability mass goes
to its unique designated signal. It remains to show that the
partitions are intervals. If not, then there would be two signals
$\SIG['] > \SIG$ and price points $i < i'$ such that
$\SigP{i}{\SIG} = \PVal{i}, \SigP{i'}{\SIG[']} = \PVal{i'}$.
Then, reallocating the probability mass \SigP{i'}{\SIG[']}
to signal \SIG instead increases the seller's revenue by at least
$\SigP{i'}{\SIG[']} \cdot (\Val{\SigI{\SIG}} - \Val{\SigI{\SIG[']}}) \geq 0$.
\end{proof}

The dynamic program for segmentation into intervals is now standard.
Let \Revenue{i}{m} denote the optimal revenue a seller can obtain from
buyer types $\SET{i,i+1,\dots, \NUMVAL}$ with $m$ signals, 
when the lowest price is \Val{i}. 
\Revenue{i}{m} satisfies the recurrence
$\Revenue{i}{0} = 0$ and 
$\Revenue{i}{m} = \max_{i < i' \leq n} 
  \big( \Revenue{i'}{m-1} + \Val{i} \cdot \sum_{j=i}^{i'-1} \PVal{j} \big)$.
The maximum attainable revenue can be found by exhaustive search
of \Revenue{i}{\NUMSIG} over all $i$.



\section{Hardness of General Persuasion}

\label{sec:hardness}

In this section, we present the proof of
Theorem~\ref{thm:hardness-intro} from the introduction. 
More accurately, the following theorem shows the hardness of maximizing \psu
  within any constant.

\begin{theorem} \label{thm:hardness}
Unless $\PP = \NP$, for any constant $c > 0$, 
there is no polynomial-time algorithm for the following problem. 
Given a Bayesian persuasion game \GAME and cardinality constraint
\NUMSIG on the number of signals, 
construct a signaling scheme \SigScheme using at most \NUMSIG signals 
such that the \psu \Util[-]{\SigScheme} under \SigScheme
is at least $c \cdot \Util[-]{\SigScheme[*]}$, where
\SigScheme[*] is the signaling scheme maximizing \Util[-]{\SigScheme}.
\end{theorem}

Because the sender utility and \psu are within a factor of 
$\frac{\NUMSIG-1}{\NUMSIG}$ of each other, this implies the same
hardness result for the sender utility, proving Theorem~\ref{thm:hardness-intro}.

We prove Theorem~\ref{thm:hardness} by establishing hardness for
a game we call the \HEGGFULL (\HEGG). 
There is a hypergraph $H=(V,E)$ which is commonly known to the sender
and receiver. 
The state of nature is a hyperedge $e^* \in E$, drawn
from the uniform distribution.

The receiver has two types of actions available: trying to guess the
hyperedge, or ``hedging her bets'' by guessing a vertex 
$v \in V$. If she guesses an edge $e$, then she gets 1 if her
guess was correct ($e=e^*$), and 0 otherwise. 
If she guesses a vertex $v$, she gets $1/d_v$ (the degree of $v$) if
$v$ is incident on $e^*$, and 0 otherwise.

The sender's utility is determined by the receiver's guess.
If the receiver guesses an edge, the sender gets utility 0,
regardless of whether the guess is correct.
If the receiver guesses a vertex $v$, the sender has utility $1/d_v$
(the same as the receiver) if $v$ is incident on $e^*$, and 0 otherwise.

Since the sender has access to $e^*$, it is his goal to design a signaling scheme
that narrows down the possible states of nature for the receiver
enough that she can get an incident vertex, but not so much as to
induce her to guess a hyperedge. 
This is accomplished by making the posterior distribution conditioned
on any signal uniform across edges incident on a particular vertex.
Ideally, we would like this to be the case for all signals, but this
may simply be impossible. However, we can achieve it for all but one
signal.

\begin{definition} \label{def:vertex-centric}
A signaling scheme \SigScheme is \emph{vertex-centric} if
for all signals \SIG except at most one, there exists a node
$v = v(\SIG)$ such that $\SigP{v}{e} =\SigP{v}{e'}$ 
for all hyperedges $e, e' \ni v$, 
and $\SigP{v}{e} = 0$ for all hyperedges $e \not\ni v$.
\end{definition}

That is, in a vertex-centric signaling scheme, all but one signal
induce a uniform posterior distribution over edges incident on one
vertex. 

\begin{lemma} \label{lem:uniform-posterior}
For any signaling scheme \SigScheme, 
there is a vertex-centric signaling scheme \SigScheme['] with
$\Util[-]{\SigScheme[']} \geq \Util[-]{\SigScheme}$, 
and which can be constructed from \SigScheme in polynomial time.
\end{lemma}

\begin{proof}
Consider any signaling scheme \SigScheme, characterized by the
probabilities \SigP{e}{\SIG} that the state of the world is $e$ and
the sender sends the signal \SIG. 
(Recall that these are not conditional probabilities.)
For each signal \SIG, there is a unique (after tiebreaking) action
that the receiver takes, either a hyperedge $e$ or a vertex $v$. 
If the receiver chooses a hyperedge $e$, the sender's utility is 0;
for a vertex $v$, it is 
$\SUtil{\SIG} = \frac{1}{d_v} \cdot \sum_{e \text{ incident on } v} \SigP{e}{\SIG}$.

Let \GSIG be the designated garbage signal; without loss of generality
(by renaming), it minimizes \SUtil{\GSIG}.
First, we may assume w.l.o.g.~that the receiver
does not choose a hyperedge for any signal $\SIG \neq \GSIG$. 
Otherwise, since the sender's utility $\SUtil{\SIG} = 0$,
we could reallocate all probability mass from \SIG to \GSIG without 
changing the \psu; under the new signal (which is never sent), the
receiver w.l.o.g.~plays a vertex.

Consider any signal $\SIG \neq \GSIG$, and let $v$ be the vertex the
receiver chooses in response to \SIG. 
First, if there is any hyperedge $e$ not incident on $v$ with
$\SigP{e}{\SIG} > 0$, we can safely lower \SigP{e}{\SIG} to 0
(reassigning the probability mass to \GSIG), without
changing the receiver's action (because $e$ was not incident on $v$,
this change cannot make $v$ less attractive), and without 
affecting \SUtil{\SIG}.

Let $d=d_v$ be the degree of $v$, and $e_1, e_2, \ldots, e_d$ the
hyperedges incident on $v$, sorted such that
$\SigP{e_1}{\SIG} \geq \SigP{e_2}{\SIG} \geq \ldots \geq \SigP{e_d}{\SIG}$.
If $\SigP{e_1}{\SIG} > \SigP{e_d}{\SIG}$, then the receiver's expected
utility from choosing $e_1$ is \SigP{e_1}{\SIG}, whereas her utility
from choosing $v$ is $\frac{1}{d} \sum_{i=1}^d \SigP{e_i}{\SIG} < \SigP{e_1}{\SIG}$.
This would contradict the receiver's playing $v$.

Notice that the changes do not affect the utility under any signal
except the garbage signal, so the \psu stays the same.
\end{proof}

Now consider the following optimization problem: find a vertex-centric
signaling scheme \SigScheme with a dedicated garbage signal \GSIG that 
maximizes the \psu $\Util[-]{\SigScheme}$.
By definition of a vertex-centric signaling scheme, 
$\SigP{e}{v} = \SigP{e'}{v}$ for all hyperedges $e,e'$ incident on $v$;
we denote this quantity by $y_v$.
Then, the probability of sending the signal inducing the receiver to
choose $v$ is $\sum_{e \ni v} \SigP{e}{v} = d_v y_v$, and the
resulting sender utility conditioned on sending it is $1/d_v$.
A vertex-centric signaling scheme is entirely determined by the
$\NUMSIG-1$ vertices and their associated probabilities $y_v$; hence,
the optimization problem can be expressed as follows.

\begin{LP}{Maximize}{\Norm[1]{\vc{y}}}
\sum_{v \in e} y_v \leq \frac{1}{\SetCard{E}} \quad \text{ for all } e \in E,\\
\Norm[0]{\vc{y}} \leq \NUMSIG-1,\\
\vc{y} \geq \vc{0}.
\end{LP}

The first constraint captures that the total probability of all
signals sent when the state of the world is $e$ can be at most the
probability that the state of the world is $e$, which is $1/\SetCard{E}$.
Rescaling all $y_v$ values by a factor $\SetCard{E}$ and removing that
constant factor from the objective gives us the following equivalent characterization.

\begin{LP}[signaling-scheme-LP]{Maximize}{\Norm[1]{\vc{y}}}
\sum_{v \in e} y_v \leq 1 \quad \text{ for all } e \in E,\\
\Norm[0]{\vc{y}} \leq \NUMSIG-1,\\
\vc{y} \geq \vc{0}.
\end{LP}

Notice that Program~\eqref{signaling-scheme-LP} would exactly be an
Independent Set characterization if the $y_v$ were restricted to be
integral. The following lemma shows that the upper bound on the
support of $\vc{y}$ is enough to ensure that the optimal solution
cannot be approximated to within any constant (when the hyperedges are
large enough). 

\begin{lemma} \label{lem:approximation-of-optimiation}
For any constant $r \geq 1$, unless $\PP = \NP$, the optimum solution of
Program~\eqref{signaling-scheme-LP} cannot be approximated to within
a factor better than $1/r$.
\end{lemma}

\begin{proof}
We give a reduction from the gap version of \IS.
Given a graph $G=(V,E)$ and any constant $\epsilon > 0$,
and the promise that the largest independent set of $G$ has size
either less than $n^\epsilon$ or more than $n^{1-\epsilon}$, 
it is \NP-hard to answer ``No'' in the former case and ``Yes'' in the
latter~\cite{hastad:clique}.
For our reduction, we specifically choose $\epsilon = \frac{1}{r+2}$.

Given $G$, we create a hypergraph $H=(V,E')$ on the same node set, 
whose hyperedges are exactly the cliques of size $r+1$ in $G$, i.e.,
$E' = \Set{S\subseteq V}{S \text{ is a clique of size } r+1 \text{ in } G}$. 
The constraint on the support size of $\vc{y}$ is $\NUMSUPP = n^{1-\epsilon}$.
Notice that the reduction is computed in time $O(n^{r+1})$, which is
polynomial in $n$ for constant $r$.
We will show that if $G$ has an independent set of size $n^{1-\epsilon}$, 
then the objective value of Program~\eqref{signaling-scheme-LP} is $\NUMSUPP$,
whereas if $G$ has no independent set of size $n^\epsilon$, then the
objective value is less than $\NUMSUPP/r$. 

First, suppose that $G$ has an independent set $S$ of size $\NUMSUPP$.
Consider the solution $\vc{y}$ to Program~\eqref{signaling-scheme-LP}
which sets $y_v = 1$ for all $v \in S$, and $y_v = 0$ for all others.
Because $S$ is independent in $G$, it contains at most one vertex from
each $(r+1)$-clique; hence, the proposed solution is valid, and it
achieves an objective value of $\NUMSUPP$.

Conversely, let $\vc{y}$ be a solution 
to Program~\eqref{signaling-scheme-LP}, and assume that its objective
value is at least $\frac{1}{r} \cdot \NUMSUPP$. 
Let $S$ be the set of all indices $v$ such that $y_v > \frac{1}{r+1}$.
Then, by the assumed lower bound on the objective value, 
$|S| \geq \frac{\NUMSUPP}{r^2}$.

Consider the subgraph $G[S]$ induced by $S$ in $G$. 
By the constraint for each hyperedge, $G[S]$ contains no $(r+1)$-clique;
otherwise, the corresponding $y_v$ would add up to more than 1.
Now, Ramsey's Theorem implies that $G[S]$ contains an independent set
of size $\Omega(n^{1/r})$, as follows.
Recall that the Ramsey Number $R(r+1,b)$ is the minimum size $s$ such
that each graph of size $s$ contains a clique of size $r+1$ or an
independent set of size $b$. 
Because $R(r+1,b) \leq \binom{r+b-1}{r} \in O(b^r)$
\cite{van2001course}, and $G[S]$ is a
graph of size at least $\frac{\NUMSUPP}{r^2}$ not containing any $(r+1)$-clique,
it must contain an independent set of size at least 
$\Omega((\NUMSUPP/r^2)^{1/r}) = \Omega(n^{(1-\epsilon)/r}) = \omega(n^{\epsilon})$.
This completes the proof.
\end{proof}

\begin{extraproof}{Theorem~\ref{thm:hardness}}
The proof is now straightforward. Given an instance of the \IS
problem, we construct the instance of the \HEGG according to the proof
of Lemma~\ref{lem:approximation-of-optimiation}, setting the allowed
number of signals to $\NUMSIG = 1+\NUMSUPP$.
Feasible solutions to Program~\eqref{signaling-scheme-LP} 
exactly capture vertex-centered signaling schemes,
and the objective value is the \psu (scaled by $\SetCard{E'}$).
\end{extraproof}



\section{Conclusion}
\label{sec:conclusion}

Our work raises several natural questions for future work.
First, while we provide a constant-factor approximation algorithm and
a QPTAS for social welfare in the pricing game, we did not actually
establish \NP-hardness. 
Is there a polynomial-time algorithm for maximizing social welfare
subject to limited communication?
Could at least a PTAS or an FPTAS be obtained?
For the more general persuasion problem with limited communication, we
establish that no approximation of social welfare to within any
constant is possible. Can this result be strengthened to logarithmic
or polynomial hardness?

In the present article, we are focusing on maximizing seller
revenue and social welfare. 
\citet{Bergemann14} also consider maximizing buyer's utility. 
It is not hard to see that given the price points for each signal, 
the buyer's utility can be maximized by a linear program very similar
to \eqref{lp:welfare-maximization}.
However, it is not clear that the overall objective function is still
submodular, or whether a similar greedy algorithm to the one from
Section~\ref{sec:greedy} optimally solves the corresponding LP.

Beyond revenue and social or buyer welfare, one could consider other
objectives for the principal. While for full generality of the
persuasion problem, our results preclude constant-factor approximation
guarantees, it would be of interest to identify other natural classes
in which limits on communication have mild consequences, and in which
good signaling schemes with limited communication can be designed
efficiently. 


\xhdr{Acknowledgments.}

This work is supported in part by NSF grant CCF-1423618. We would also
like to thank Alex Eager for useful conversations and anonymous
reviewers for their constructive comments and suggestions.

\bibliographystyle{plainnat}

\bibliography{../davids-bibliography/names,../davids-bibliography/conferences,../davids-bibliography/bibliography,../davids-bibliography/publications,../bibliography/paper-specific,../bibliography/shaddin}


 \appendix
 \section{Proof of Theorem~\ref{thm:greedy-optimal}}
\label{sec:greedy-proof}
In this section, we prove Theorem~\ref{thm:greedy-optimal}.

\begin{rtheorem}{Theorem}{\ref{thm:greedy-optimal}}
The algorithm \algoname{Construct-Signaling-Scheme} solves 
the linear program \eqref{lp:welfare-maximization} optimally.
\end{rtheorem}

\begin{proof}
Let \SigScheme[*] be an optimal solution for the linear program
\eqref{lp:welfare-maximization}, and \SigScheme the signaling scheme
constructed by the algorithm \algoname{Construct-Signaling-Scheme}.
We will show that \SigScheme[*] can be (gradually) transformed into
\SigScheme without decreasing its solution quality, proving optimality
of \SigScheme.

First, we may assume without loss of generality that 
$\SigP[*]{i}{\SIG} = 0$ for all $i < \SigI{\SIG}$,
since setting them to 0 affects neither the objective value nor the
constraints.
Let \SigC{\SIG}, \SigC[*]{\SIG} denote column \SIG of \SigScheme,
\SigScheme[*], i.e., the vector of probabilities that constitute
signal \SIG.

Assume that $\SigScheme \neq \SigScheme[*]$.
Let \SIG be the smallest index (i.e., with largest price) such
that $\SigC{\SIG} \neq \SigC[*]{\SIG}$.
Let $i$ be minimal such that $\SigP{i}{\SIG} \neq \SigP[*]{i}{\SIG}$.
For notational convenience, since we will mostly focus on 
the signal \SIG, we write 
$\OSVec = \SigC{\SIG}$ and $\OSVec[*] = \SigC[*]{\SIG}$.
Let $\RPVec = \PVec - \sum_{\SIG['] < \SIG} \SigC{\SIG[']}
= \PVec - \sum_{\SIG['] < \SIG} \SigC[*]{\SIG[']}$ 
be the vector of residual probabilities at the time that signal \SIG was
greedily constructed. 
We now distinguish two cases:

\xhdr{Case 1: $\OS{i} < \OS[*]{i}$:}
By Lemma~\ref{lem:maxtail}, $\TS{i} \geq \TS[*]{i}$, 
and because $\OS{i} < \OS[*]{i}$, we get that $\TS{i+1} > \TS[*]{i+1}$.
In particular, there must be an index $i' > i$ such that
$\OS{i'} > \OS[*]{i'}$; let $i'$ be the smallest such index.
Let $\delta = \min (\OS{i'} - \OS[*]{i'}, \OS[*]{i} - \OS{i}) > 0$.

We next show that under \SigScheme[*], \emph{all} signals combined
must use all probability mass of type $i'$, i.e.,
$\sum_{\SIG['] = 1}^{\NUMSIG} \SigP[*]{i'}{\SIG[']} = \PVal{i'}$.
If this were not the case, then define 
$\epsilon = \min (\delta, \PVal{i'} - \sum_{\SIG['] = 1}^{\NUMSIG} \SigP[*]{i'}{\SIG[']})
> 0$, and consider modifying \SigScheme[*] by updating
$\OS[']{i'} = \OS[*]{i'} + \epsilon$ and
$\OS[']{i} = \OS[*]{i} - \epsilon$ 
(and leaving $\OS[']{j} = \OS[*]{j}$ for all $j \neq i, i'$).
By choice of $\epsilon$, this new solution \SigScheme['] does not
violate the non-negativity or total probability constraints, and we claim that 
(1) it satisfies the revenue constraint \eqref{eqn:tailsum-inequality}
for all $j$, and 
(2) its welfare is strictly higher than that of \SigScheme[*].
        
To check the revenue constraints, notice first that the seller's
revenue under \SigScheme['] for indices $j \leq i$ and $j > i'$ is
unchanged, so \eqref{eqn:tailsum-inequality} still holds for such
$j$. It remains to consider $j \in \SET{i+1, \ldots, i'}$.
Fix one such $j$.
By definition of $i$ and $i'$, we have that $\OS{j'} \leq \OS[']{j'}$ for
$\SigI{\SIG} \leq j' < i'$; in particular, we can infer for our $j$
that

\begin{align}
\PMY[']{[\SigI{\SIG},j)} 
& = \sum_{j'=\SigI{\SIG}}^{j-1} \OS[']{j'}
\; \geq \; \sum_{j'=\SigI{\SIG}}^{j-1} \OS{j'}
\; = \; \PMY{[\SigI{\SIG},j)}.
\label{eqn:tailsums-after-modification}
\end{align}

Because $\TS{i} \geq \TS[*]{i}$ and $\OS{i} \leq \OS[*]{i} - \epsilon$,
and the definition of $i'$, we get that for all $j \in \SET{i+1, \ldots, i'}$,
\begin{align}
\TS{j} & \geq \TS[*]{j} + \epsilon \; = \; \TS[']{j}.
\label{eqn:intermediate-tailsums}
\end{align}
Combining these two inequalities with the fact that \SigScheme
satisfies all revenue constraints, in particular 
$\Val{j} \cdot \TS{j} \leq \Val{\SigI{\SIG}} \cdot \TS{\SigI{\SIG}}$, 
we get that

\begin{align*}
\Val{\SigI{\SIG}} \cdot \PMY[']{[\SigI{\SIG},j)}
& \stackrel{\eqref{eqn:tailsums-after-modification}}{\geq}     
\Val{\SigI{\SIG}} \cdot \PMY{[\SigI{\SIG},j)}  
\; \geq \; (\Val{j} - \Val{\SigI{\SIG}}) \cdot \TS{j} 
\; \stackrel{\eqref{eqn:intermediate-tailsums}}{\geq} \;  
(\Val{j} - \Val{\SigI{\SIG}}) \cdot \TS[']{j}.
\end{align*}

Adding $\Val{\SigI{\SIG}} \cdot \TS[']{j}$ to both sides now shows
that \SigScheme['] satisfies the revenue inequality for $j$.
However, notice that the objective value has increased by
$\epsilon(\Val{i'} - \Val{i})$, contradicting the optimality of
\SigScheme[*]. Hence, we have shown that
$\sum_{\SIG['] = 1}^{\NUMSIG} \SigP[*]{i'}{\SIG[']} = \PVal{i'}$.
We will show that there is another signal $\SIG['] > \SIG$ such that
we can redistribute probability mass between signals $\SIG, \SIG[']$
without affecting either the objective or the constraints, while
making \SigScheme[*] and \SigScheme more similar.

Because $\OS{i'} > \OS[*]{i'}$ and
$\sum_{\SIG['] = 1}^{\NUMSIG} \SigP[*]{i'}{\SIG[']} = \PVal{i'}$,
there must be a signal $\SIG['] > \SIG$ such that
$\SigP[*]{i'}{\SIG[']} > \SigP{i'}{\SIG[']}$.
Fix \SIG['] to be the smallest such index, and define
$\epsilon = \min (\SigP[*]{i'}{\SIG[']}, \delta)$.
Consider the modified signaling scheme \SigScheme['] with
\begin{align*}
\OS[']{i} & = \OS[*]{i} - \epsilon, \\
\OS[']{i'} & = \OS[*]{i'} + \epsilon, \\
\SigP[']{i}{\SIG[']}  & = \SigP[*]{i}{\SIG[']} + \epsilon, \\
\SigP[']{i'}{\SIG[']} & = \SigP[*]{i'}{\SIG[']} - \epsilon, \\
\SigP[']{j}{\SIG['']} & = \SigP[*]{j}{\SIG['']} \quad \text{ for all
  other } j, \SIG[''].
\end{align*}

Because this assignment only redistributes probability mass,
the probability mass constraints cannot be violated, and the social
welfare stays the same.
Non-negativity follows from the choice of $\epsilon$.
That \OSVec['] satisfies the revenue constraints follows exactly the
same proof as in the previous definition of \OSVec['], because the
adjustment is the same.
Finally, for all $j \in \SET{i+1, \ldots, i'}$, the tail sums of
probabilities have decreased under \SIG['], implying that so has the
revenue for those prices. This means that all revenue constraints are
also satisfied for \SIG['].

This modification makes \OS{i} and \OS[*]{i} more similar, and
repeating this procedure with another signal \SIG['] as needed, they
will eventually become the same.

\xhdr{Case 2: $\OS{i} > \OS[*]{i}$:}
Analogous to the first case, we will begin by showing that 
$\sum_{\SIG['] = 1}^{\NUMSIG} \SigP[*]{i}{\SIG[']} = \PVal{i}$.
For contradiction, assume that 
$\sum_{\SIG['] = 1}^{\NUMSIG} \SigP[*]{i}{\SIG[']} < \PVal{i}$, and 
let $\epsilon = \min (\OS{i} - \OS[*]{i},
\PVal{i} - \sum_{\SIG['] = 1}^{\NUMSIG} \SigP[*]{i}{\SIG[']})$.
Define an improved signaling scheme by setting
$\OS[']{i} = \OS[*]{i} + \epsilon$, and 
$\OS[']{j} = \OS[*]{j}$ for all $j \neq i$.

The probability mass constraints are still satisfied, and nothing has
changed for indices $j > i$, so those revenue constraints are still
satisfied. Next, consider an index $j \in \SET{\SigI{\SIG}+1, \ldots, i}$.
Because $i$ is the first index with $\OS{i} \neq \OS[*]{i}$, we get that
$\PMY[*]{[\SigI{\SIG},j)}  = \PMY{[\SigI{\SIG},j)} $.

By Lemma~\ref{lem:maxtail}, 
$\TS{i+1} \geq \TS[*]{i+1}$, and because
$\OS{i} \geq \OS[*]{i} + \epsilon$, we obtain that
$\TS{j} \geq \TS[*]{j} + \epsilon = \TS[']{j}$.
Combining this and the previous inequality with the 
revenue constraint for \OSVec at index $j$ (namely, that
$\Val{\SigI{\SIG}} \cdot \TS{\SigI{\SIG}} \geq \Val{j} \cdot \TS{j}$),
we conclude --- analogously to the previous case --- that 
\begin{align*}
\Val{\SigI{\SIG}} \cdot \PMY[']{[\SigI{\SIG},j)}
& =  \Val{\SigI{\SIG}} \cdot \PMY{[\SigI{\SIG},j)} 
\; \geq \; (\Val{j} - \Val{\SigI{\SIG}}) \cdot \TS{j} 
\; \geq \; (\Val{j} - \Val{\SigI{\SIG}}) \cdot \TS[']{j}.
\end{align*}

Adding $\Val{\SigI{\SIG}} \cdot \TS[']{j}$ to both sides now
establishes that the revenue constraint is satisfied at index $j$.
Because the objective value strictly increased, we obtain a
contradiction to the optimality of \SigScheme[*].
Hence, from now on, we assume that
$\sum_{\SIG['] = 1}^{\NUMSIG} \SigP[*]{i}{\SIG[']} = \PVal{i}$.

As in Case 1, there must be some signal $\SIG['] > \SIG$ with
$\SigP[*]{i}{\SIG[']} > 0$. Our goal is again to increase
\OS[*]{i} to make it closer to \OS{i}, and do so by reassigning
probability mass from signal \SIG['].
However, in this case, doing so involves a more careful reallocation
of probability mass to ensure all revenue constraints are satisfied.
In fact, we will use the algorithm \algoname{Construct-One-Signal}
to construct a vector $\vc{d}$ describing the probability
reallocation, and then set 
\begin{align*}
\OSVec['] & = \OSVec[*] + \vc{d}, \\
\SigC[']{\SIG[']} & = \SigC[*]{\SIG[']} - \vc{d}.
\end{align*}

The ``residual probability vector'' $\vc{r}$ in this case is defined as
$r_i = \min (\SigP[*]{i}{\SIG[']}, \OS{i} - \OS[*]{i})$ and
$r_j = \SigP[*]{j}{\SIG[']}$ for $j \neq i$.
It captures the fact that we can at most reallocate all of the
probability mass of \SigC[*]{\SIG[']}, but also must ensure that the
new signal satisfies $\OS[']{i} \leq \OS{i}$.
Since the modified version of signal \SIG must satisfy the revenue
constraints with target price \Val{\SigI{\SIG}}, we make this price
the target of the construction of $\vc{d}$, by setting
$u_i = \Val{\SigI{\SIG}}$, 
$u_j = \Val{j}$ for $j > i$, 
and $u_j = 0$ for $j < i$.
Then, the change vector is defined as
$\vc{d} = \algoname{Construct-One-Signal}(\vc{r}, \vc{u}, i)$. 
We now want to show that the signals
$\OSVec['], \SigC[']{\SIG[']}$ defined above satisfy all constraints.

First, because probability mass only gets moved around between
signals, the welfare stays the same, and the probability and
non-negativity constraints cannot get violated because 
\algoname{Construct-One-Signal} at most uses $r_j$ units of
probability in coordinate $j$.
We therefore focus on verifying the revenue constraints.

We begin with \SigC[']{\SIG[']}, which only saw its probabilities decrease.
The seller's revenue at the target index $\SigI{\SIG[']}$ decreased by
$\Val{\SigI{\SIG[']}} \cdot d_{[i,\NUMVAL]}$.
Consider any index $j > \SigI{\SIG[']}$.
The seller's revenue at index $j$ decreased by $\Val{j} \cdot d_{[j,\NUMVAL]}$.
If $j \notin I$ at the termination of the algorithm,
then $d_j = r_j$, and $\SigP[']{j}{\SIG[']} = 0$, meaning that
price $j$ is no more attractive than $j+1$ or $j-1$ to the seller.
Otherwise, Lemma~\ref{lem:revenue-tight} implies that
\begin{align*}
\Val{j} \cdot d_{[j,\NUMVAL]} 
& = u_j \cdot d_{[j,\NUMVAL]}
\; \stackrel{\text{Lemma~\ref{lem:revenue-tight}}}{=} \; u_i \cdot d_{[i,\NUMVAL]}
\; = \; \Val{\SigI{\SIG}} \cdot d_{[i,\NUMVAL]}
\; \geq \; \Val{\SigI{\SIG[']}} \cdot d_{[i,\NUMVAL]}.
\end{align*}
In particular, $j$ cannot have become more attractive to the seller
than \SigI{\SIG[']}.

We next verify that the revenue constraints are also satisfied for the
signal \OSVec[']. Here, the seller's revenue for price point $j$
increased by $\Val{j} \cdot d_{[j,\NUMVAL]}$.
For all $j > i$, the algorithm \algoname{Construct-One-Signal} ensures that 
$\Val{j} \cdot d_{[j,\NUMVAL]} 
= u_j \cdot d_{[j,\NUMVAL]} 
\leq u_i \cdot d_{[i,\NUMVAL]} 
= \Val{\SigI{\SIG}} \cdot d_{[i,\NUMVAL]}$,
so the increase in revenue is no larger for price point $j$ than for
\SigI{\SIG}.
Thus, we have obtained the revenue constraint
$\Val{\SigI{\SIG}} \cdot \TS[']{\SigI{\SIG}} \geq \Val{j} \cdot \TS[']{j}$.
By rewriting $\TS[']{\SigI{\SIG}} = \TS[']{j} + (\TS[']{\SigI{\SIG}} - \TS[']{j})$
and rearranging, we obtain the useful form
\begin{align}
\TS[']{j} & \leq
\frac{\Val{\SigI{\SIG}} \cdot \PMY[']{[\SigI{\SIG},j)}}{%
\Val{j} - \Val{\SigI{\SIG}}}.
\label{eqn:tailsum-bound-large}
\end{align}

The slightly tricky part is the indices
$j \in \SET{\SigI{\SIG} + 1, \ldots, i}$.
While the probabilities \OS[']{j} for $j < i$ do not increase,
the tail probabilities \TS[']{j} do by virtue of increases in
\OS[']{j} for $j \geq i$; hence, we need to also consider these price points.
Fix such a $j \in \SET{\SigI{\SIG} + 1, \ldots, i}$.
We will first show that $\TS[']{j} \leq \TS{j}$.

Let $i' > i$ be a smallest index with
$\OS{i'} < \PVal{i'} - \sum_{\SIG[''] < \SIG} \SigP{i'}{\SIG['']}$.
(If no such $i'$ exists, then let $i' = \NUMVAL + 1$.)
We will first show that $\TS[']{i'} \leq \TS{i'}$.
This holds trivially when $i' = \NUMVAL + 1$.
Otherwise, we apply Lemma~\ref{lem:revenue-tight} to the construction
of \OSVec, and obtain that that 
$\TS{i'} = \frac{\Val{\SigI{\SIG}} \cdot \PMY{[\SigI{\SIG},i')}}{%
\Val{i'} - \Val{\SigI{\SIG}}}$.
Now, notice that for all $j' \in \SET{\SigI{\SIG}, \ldots, i'-1}$, 
we have that $\OS[']{j'} \leq \OS{j'}$, for different reasons.
\begin{enumerate}
\item For $j' > i$, this follows because the definition of $i'$
  implies that
$\OS{j'} = \PVal{j'} - \sum_{\SIG[''] < \SIG} \SigP{j'}{\SIG['']}$
is as large as it can possibly be.
\item For $j' = i$, it follows because
$\OS[']{i} = \OS[*]{i} + \epsilon \leq \OS{i}$.
\item For $j' < i$, it follows because
$\OS[']{j'} = \OS[*]{j'} = \OS{j'}$.
\end{enumerate}
This implies that 
$\PMY[']{[\SigI{\SIG},i')} \leq \PMY{[\SigI{\SIG},i')}$,
and hence --- using Inequality~\eqref{eqn:tailsum-bound-large} ---
that $\TS[']{i'} \leq \TS{i'}$.
Finally, the previous three cases show that for the fixed $j$,
\begin{align*}
\TS[']{j} 
& = \TS[']{i'} + \sum_{j' = j}^{i' - 1} \OS[']{j'}
\; \leq \; \TS{i'} + \sum_{j' = j}^{i' - 1} \OS{j'}
\; = \; \TS{j}.
\end{align*}

Having shown that $\TS[']{j} \leq \TS{j}$, we next apply
Lemma~\ref{lem:revenue-tight} to \OSVec at price point $j$
to obtain that
$\TS{j} \leq \frac{\Val{\SigI{\SIG}} \cdot \PMY{[\SigI{\SIG},j)} }{%
\Val{j} - \Val{\SigI{\SIG}}}$.
Combining this with the fact that
$\PMY{[\SigI{\SIG},j)} = \PMY[']{[\SigI{\SIG},j)}$
for $j < i$ by the third case of the above case distinction,
we obtain that
$\TS[']{j} \leq \frac{\Val{\SigI{\SIG}} \cdot \PMY[']{[\SigI{\SIG},j)}}{%
\Val{j} - \Val{\SigI{\SIG}}}$,
which is an equivalent way of rewriting the revenue constraint.

Again, this modification makes \OS{i} and \OS[*]{i} more similar, and
repeating this procedure with additional signals \SIG['] as needed, they
will eventually become the same.
\end{proof}

 \section{Proof of Lemma~\ref{lem:small-increase-welfare}}
\label{sec:submodularity-proofs}

In this section, we provide a proof of
Lemma~\ref{lem:small-increase-welfare}, restated here for convenience.

\begin{rtheorem}{Lemma}{\ref{lem:small-increase-welfare}}
If $S \subseteq T$, then for any $k, B, \epsilon$: 
$\PWelfare[(k,B+\epsilon)]{T} - \PWelfare[(k,B)]{T}
\leq \PWelfare[(k,B+\epsilon)]{S} - \PWelfare[(k,B)]{S}$.
\end{rtheorem}

We will prove this lemma for sufficiently small
$\epsilon$, which allows us to couple the executions tightly; the
inequalities can then be added to imply the lemma for arbitrary
$\epsilon$.

By comparing the solutions to the linear
program~\eqref{lp:welfare-maximization}, we can ensure that any
constraint that becomes tight in the solution for set $T$ with
bound $B + \epsilon$, but is not tight with bound $B$ (and similarly
for $S$) would not have become tight for any $\epsilon' < \epsilon$.
This will localize the changes, and use revenue indifference for the
seller.
By summing over all such iterations (there will only be finitely
many, because $\epsilon$ is chosen so that at least one more
constraint becomes tight), we eventually prove the lemma.

In the analysis, we are interested in four different signaling
schemes, constructed by \algoname{Construct-Signaling-Scheme} when run
with different sets of price points and upper bounds $B$.
We will assume here that $k \in S \subseteq T$.
Specifically we define:

\begin{itemize}
\item $\PMS = (\pmS{\SIG}{i})_{\SIG,i}$: the probability mass of
  type $i$ assigned to signal \SIG when the algorithm is run with
  price point set $S$ and an upper bound of $B$.
\item $\PMSP = (\pmSp{\SIG}{i})_{\SIG,i}$: probability mass for price
  point set $S$ and an upper bound of $B+\epsilon$.
\item $\PMT = (\pmT{\SIG}{i})_{\SIG,i}$: probability mass for price
  point set $T$ and an upper bound of $B$.
\item $\PMTP = (\pmTp{\SIG}{i})_{\SIG,i}$: probability mass for price
  point set $T$ and an upper bound of $B+\epsilon$.
\end{itemize}

To avoid notational confusion, we will use 
\SIGS to denote signals under the signaling scheme for $S \cup \SET{k}$
and \SigIS{\SIG} to denote their price points;
signals under the signaling scheme for $T \cup \SET{k}$ are denoted by
\SIGT, and their price points by \SigIT{\SIGT}.

Since we are interested in the change in the signaling scheme as we
increase the bounds from $B$ to $B+\epsilon$, we define
$\dS{\SIGS}{i} = \pmS{\SIGS}{i} - \pmSp{\SIGS}{i}$ and
$\dT{\SIGT}{i} = \pmT{\SIGT}{i} - \pmTp{\SIGT}{i}$.

In order to understand the \dS{\SIG}{i} better, consider the effect of
changing the total probability constraint for the signal with price
point $k$ from $B$ to $B+\epsilon$. When the signal with price point
$k$ is constructed, we may now add some more probability mass for
types $i \geq k$. 
Subsequently, by Lemma~\ref{lem:pointwise-dominance}, the algorithm
continues with less (or equal) residual probability mass for all
types. This might in turn mean that for later signals \SIG with price
points $k' < k$, the probability mass for some type $i$ might get used
up earlier. In turn, this will speed up the addition of probability
mass for types $i' \in \SET{k'+1, \ldots, i-1}$.
In a sense, what happens is that the additional probability mass for
the signal with price point $k$ ``displaces'' some of that mass from
other signals, increasing different probability masses.

Consider some signal $\SIG \in \SIGSET$.
Let $\LSig{\SIGSET}{\SIG} = \Set{\SIG['] \in \SIGSET}{\SIG['] \leq \SIG}$ be the set of
signals constructed up to \SIG. 
We are interested in which price points may change their overall
allocation of probability in the signals constructed up to \SIG as a
result of the increase from $B$ to $B+\epsilon$.
Notice that the candidates are only those that did not have their
probability mass already used up when the construction reached \SIG
with an upper bound of $B$. 
To capture this, we define
$\SigI{\SIG} < \nfS[\SIG]{1} < \nfS[\SIG]{2} < \ldots < \nfS[\SIG]{\nfSn[\SIG]}$ 
to be the indices which still had probability mass available
after signal \SIG was constructed, i.e., such that
$\opmS{\LSig{\SIGSET}{\SIG}}{\nfS[\SIG]{j}} <
\PVal{\nfS[\SIG]{j}}$.
For notational convenience, we define
$\nfS[\SIG]{\nfSn[\SIG]+1} = \NUMVAL+1$ and $\Val{\NUMVAL+1} = \infty$.
Similarly, define
$\SigIT{\SIGT} < \nfT[\SIGT]{1} < \nfT[\SIGT]{2} < \ldots < \nfT[\SIGT]{\nfTn[\SIGT]}$
to be the indices with
$\opmT{\LSig{\SIGSETT}{\SIGT}}{\nfT[\SIGT]{j}} < \PVal{\nfT[\SIGT]{j}}$.

We are particularly interested in indices whose overall total
probability mass (at the end of \algoname{Construct-Signaling-Scheme}) 
can increase. For ease of notation, we therefore define
$\SIG[^S] = \max \Set{\SIG}{k_{\SIG}\in S},\SIGT[^T] = \max \Set{\SIGT}{\SigIT{\SIGT} \in T}$,
and $\nfS{j} = \nfS[{\SIG[^S]}]{j}, \nfT{j} = \nfT[{\SIGT[^T]}]{j}$;
note that this implies $\pmS{\SIGSETS}{\nfS{j}} < \PVal{i}$
for all $j$, and similarly for \nfT{j}.
One type $i < k$ could also see an increase, namely, if the
displaced probability by an increase for a type $i' > k$ which
eventually became saturated causes an increase in probability mass for
some later signal with target $k' < k$. 
The only such target types would be 
$\nfS{0} = \max \Set{i \leq k}{\opmS{\SIGSETS}{i} < \PVal{i}}$
and 
$\nfT{0} = \max \Set{i \leq k}{\opmT{\SIGSETT}{i} < \PVal{i}}$,
respectively.

Notice that by making $\epsilon$ small enough, we can ensure that
$\opmSp{\LSig{\SIGSETS}{\SIG}}{\nfS[\SIG]{j}} < \PVal{\nfS[\SIG]{j}}$
  and $\opmTp{\LSig{\SIGSETT}{\SIGT}}{\nfT[\SIGT]{j}} < \PVal{\nfT[\SIGT]{j}}$
for all indices \nfS[\SIG]{j}, \nfT[\SIGT]{j} and signals $\SIGS,\SIGT$. 

Define $\epsilon$ to be the supremum of all such values.
This choice of $\epsilon$ ensures that at least one of the above
inequalities becomes tight, but for any 
$\epsilon' < \epsilon$, we have 
$\opmSp{\LSig{\SIGSETS}{\SIG}}{\nfS[\SIG]{j}} < \PVal{\nfS[\SIG]{j}}$
and $\opmTp{\LSig{\SIGSETT}{\SIGT}}{\nfT[\SIGT]{j}} < \PVal{\nfT[\SIGT]{j}}$. 
With the chosen $\epsilon$, for the index (or indices) where the
inequality becomes tight, we still have a tight revenue constraint,
because in the execution of \algoname{Construct-One-Signal}, the index
was removed from $I$ at the same time as $\SigI{\SIG}$, in the last
round of the iteration.

By choosing the proper $\epsilon$, at least one of 
$\nfS[\SIG]{j}$ and $\nfT[\SIGT]{j}$ is removed from the index set
for the next larger value $B' = B+\epsilon$. 
Because there are only finitely many candidate indices, finitely many
updates will reach the $B$ such that $\pmT{\SIGT[k]}{k} = \PVal{k}$
when running with a bound of $B$. Beyond that value of $B$, the signal
\SIGT[k] cannot be further raised.
 
As we argued above, the increase in probability mass for the signal
with price index $k$ will reduce the probability mass available for
other signals constructed subsequently in the algorithm
\algoname{Construct-Signaling-Scheme}. 
Let \SIG be such a signal, with price point $\SigI{\SIG} < k$.
A lack of available probability mass for low-value buyer types when
\SIG is constructed may make its target price \Val{\SigI{\SIG}} less
attractive to the seller; to compensate, the signal must reduce the
amount of probability mass for high-value buyers it uses.
The following lemma captures the necessary reduction.

\begin{lemma} \label{lem:compensate-for-loss}
Fix some signal \SIG, 
and let $i < i', j < j'$ be indices such that each of 
$i, i', j, j'$ is either equal to $k'$ or to one of the
\nfS[\SIG]{j''}. Then,

\begin{align} 
\frac{\dS{\SIG}{[j, j')}}{\inversediffi{j}{j'}}
& = \frac{\dS{\SIG}{[i,i')}}{\inversediffi{i}{i'}}.
\label{eqn:keep-seller-indifferent}
\end{align}

An analogous characterization holds for the
\nfT[\SIGT]{j''} and \dT{\SIGT}{[j, j')}.
\end{lemma}

\begin{proof}
Let $k'=\SigI{\SIG}$.
Because $\opmS{\LSig{\SIGSET}{\SIG}}{\nfS[\SIG]{j}} < \PVal{\nfS[\SIG]{j}}$,
we can apply Lemma~\ref{lem:revenue-tight} to the iteration in which \SIG
was constructed, and infer that for all $j \leq \nfSn[\SIG]$,
\begin{align*}
\Val{k'} \cdot \pmS{\SIG}{[k',\NUMVAL]}  
& = \Val{\nfS[\SIG]{j}} \cdot \pmS{\SIG}{[\nfS[\SIG]{j},\NUMVAL]}, 
&
\Val{k'} \pmSp{\SIG}{[k',\NUMVAL]}  
& = \Val{\nfS[\SIG]{j}} \cdot \pmSp{\SIG}{[\nfS[\SIG]{j},\NUMVAL]},
\end{align*}
whence 
$\Val{k'} \cdot \dS{\SIG}{[k',\NUMVAL]} 
= \Val{\nfS[\SIG]{j}} \cdot \dS{\SIG}{[\nfS[\SIG]{j},\NUMVAL]}$ follows.

\begin{align*}
\dS{\SIG}{[\nfS[\SIG]{j}, \nfS[\SIG]{j'})}
& =   \dS{\SIG}{[\nfS[\SIG]{j},\NUMVAL]}
    - \dS{\SIG}{[\nfS[\SIG]{j'}, \NUMVAL]}
\; = \; \frac{\Val{k'}}{\Val{\nfS[\SIG]{j}}} \cdot \dS{\SIG}{[k',\NUMVAL]} 
      - \frac{\Val{k'}}{\Val{\nfS[\SIG]{j'}}} \cdot \dS{\SIG}{[k',\NUMVAL]} 
\; = \; \big( \frac{1}{\Val{\nfS[\SIG]{j}}} - \frac{1}{\Val{\nfS[\SIG]{j'}}} \big)
      \cdot \Val{k'} \cdot \dS{\SIG}{[k',\NUMVAL]} 
\end{align*}

Hence, $\frac{\dS{\SIG}{[\nfS[\SIG]{j}, \nfS[\SIG]{j'})}}{%
\inversediffi{\nfS[\SIG]{j}}{\nfS[\SIG]{j'}}} 
= \Val{k'} \cdot \dS{\SIG}{[k',\NUMVAL+1)}$ for all 
$j < j' \leq \nfSn[\SIG]+1$, and it is easy to see that this
calculation applies for $j = k'$ as well.
\end{proof}

\begin{lemma} \label{lem:subsequence}
The $\nfS{j}, \nfT{j}, \nfS[\SIG]{j}, \nfT[\SIGT]{j}$ satisfy the
following subsequence properties.
\begin{enumerate}
\item Each index \nfT{j} also appears as an element in the sequence
  of \nfS{j}. 
\label{lem:subsequence:unindexed}
\item Let $k'$ be a price point and $\SIGS, \SIGT$ signals such that
\SIGS has price point $k'$ under \PMS and \SIGT has price point $k'$
under \PMT.
Then, each index \nfT[{\SIGT}]{j} also appears as
an element in the sequence of \nfS[\SIG]{j}.
\label{lem:subsequence:XvsY}
\item If $\SIG > \SIG[']$, then each index
$\nfS[\SIG]{j} \geq \SigIS{\SIG[']}$
also appears as an element in the sequence of \nfS[{\SIG[']}]{j}.
Similarly for  \nfT[\SIGT]{j} and \nfT[{\SIGT[']}]{j}.
\label{lem:subsequence:differentJ}
\end{enumerate}
\end{lemma}

\begin{emptyproof}
Let \SIGSET[S], \SIGSET[T] denote the sets of signals corresponding to
the price point sets $S, T$, respectively.
\begin{enumerate}
\item By Lemma~\ref{lem:pointwise-dominance}, applied to 
$S \subseteq T$, we get that
$ \pmS{\SIGSET[S]}{i} \leq \pmT{\SIGSET[T]}{i}$, 
so whenever 
$ \pmT{\SIGSET[T]}{i} < \PVal{\SIG}$, 
we also have
$ \pmS{\SIGSET[S]}{i} < \PVal{\SIG}$.
\item Analogous to the first part after applying
Lemma~\ref{lem:pointwise-dominance} with
$\SIGSET = \LSig{\SIGSET[S]}{\SIG}$ and 
$\SIGSETT= \LSig{\SIGSET[T]}{\SIGT}.$ The corresponding price
point sets are $\Set{\SigIS{\SIG} \in S}{\SigIS{\SIG}\ge k'}
\subseteq \Set{\SigIT{\SIGT}\in T}{\SigIT{\SIGT}\ge k'}$, respectively.
\item Analogous to the first part after applying
Lemma~\ref{lem:pointwise-dominance} with 
$\SIGSET = \LSig{\SIGSET[S]}{\SIG[']}$ and $\SIGSET['] = \LSig{\SIGSET[S]}{\SIG}$. \QED
\end{enumerate}
\end{emptyproof}

\begin{lemma} \label{lem:non-negative-delta}
Let $\SIG['] > \SIG[k]$ be a signal constructed after the one with
price point $k$.
Then, the total probability mass under signal \SIG[']
for the ``initial saturated segment''
$[\SigIS{\SIG[']}, \nfS[{\SIG[']}]{1})$
cannot increase when $B$ is raised to $B + \epsilon$. 
That is, $\dS{\SIG[']}{[\SigIS{\SIG[']}, \nfS[{\SIG[']}]{1})} \geq 0$.
If furthermore, $\nfS[{\SIG[']}]{1} \leq k$, 
then $\dS{\SIG[']}{[\SigIS{\SIG[']},\nfS[{\SIG[']}]{1})} = 0$.
An analogous statement holds for
$\dT{\SIGT[']}{[\SigIT{\SIGT[']}, \nfT[{\SIGT[']}]{1})}$ in place of
$\dS{\SIG[']}{[\SigIS{\SIG[']}, \nfS[{\SIG[']}]{1})}$.
\end{lemma}

\begin{proof}
By definition of \nfS[{\SIG[']}]{1}, all the probability mass for
indices $i \in [\SigIS{\SIG[']}, \nfS[{\SIG[']}]{1})$ is used up by
signals $1, \ldots, \SIG[']$ under both \PMS and \PMSP. Hence,
\begin{align*}
\sum_{\SIG \leq \SIG[']} \pmS{\SIG}{[\SigIS{\SIG[']},\nfS[{\SIG[']}]{1})} 
& =  \PVal{[\SigIS{\SIG[']},\nfS[{\SIG[']}]{1})}
\; = \;
\sum_{\SIG \leq \SIG[']} \pmSp{\SIG}{[\SigIS{\SIG[']},\nfS[{\SIG[']}]{1})}.
\end{align*}
Taking the difference and solving gives us that
$\dS{\SIG[']}{[\SigIS{\SIG[']},\nfS[{\SIG[']}]{1})} 
= - \sum_{\SIG < \SIG[']}
\dS{\SIG}{[\SigIS{\SIG[']},\nfS[{\SIG[']}]{1})}
\ge 0$ by monotonicity (Lemma~\ref{lem:pointwise-dominance}).

We prove the second part of the lemma by induction over
$\SIG['] > \SIG[k]$.
For the base case, let \SIG['] be the least \SIG such that 
$\nfS[\SIG]{1} \leq k$.
Consider a signal $\SIG < \SIG[']$.
If $\SigIS{\SIG} \leq \nfS[{\SIG[']}]{1}$, then 
part~\ref{lem:subsequence:differentJ} of Lemma~\ref{lem:subsequence}
would imply that \nfS[{\SIG[']}]{1} also appears as a
\nfS[\SIG]{j} for some $j$; in particular, it would imply that 
$\nfS[\SIG]{1} \leq \nfS[{\SIG[']}]{1} \leq k$, contradicting the
minimality of \SIG['].
Hence, $\SigIS{\SIG} > \nfS[{\SIG[']}]{1}$, meaning that no
probability mass is allocated to signals $\SIG < \SIG[']$ for types
$[\SigIS{\SIG}, \nfS[{\SIG[']}]{1})$.
Because by definition, all probability mass for such types is used up
by signals up to \emph{and including} signal \SIG['], we conclude that
$\pmS{\SIG[']}{[\SigIS{\SIG[']}, \nfS[{\SIG[']}]{1})} =
\PVal{[\SigIS{\SIG[']}, \nfS[{\SIG[']}]{1})} = 
\pmSp{\SIG[']}{[\SigIS{\SIG[']}, \nfS[{\SIG[']}]{1})}$,
and therefore
$\dS{\SIG[']}{[\SigIS{\SIG[']}, \nfS[{\SIG[']}]{1})} = 0$.

For the induction step, we use our result from the first part that
$\dS{\SIG[']}{[\SigIS{\SIG[']},\nfS[{\SIG[']}]{1})} 
= - \sum_{\SIG < \SIG[']}
\dS{\SIG}{[\SigIS{\SIG[']},\nfS[{\SIG[']}]{1})}$.
We will show that $\dS{\SIG}{[\SigIS{\SIG[']},\nfS[{\SIG[']}]{1})} = 0$
for all $\SIG < \SIG[']$. 
First, when $\SigIS{\SIG} > \nfS[{\SIG[']}]{1}$, we get that
$\pmS{\SIG}{[\SigIS{\SIG[']},\nfS[{\SIG[']}]{1})}
= \pmSp{\SIG}{[\SigIS{\SIG[']},\nfS[{\SIG[']}]{1})} = 0$, 
implying that $\dS{\SIG}{[\SigIS{\SIG[']},\nfS[{\SIG[']}]{1})} = 0$.
Otherwise, when $\SigIS{\SIG} \leq \nfS[{\SIG[']}]{1}$,
we first observe that 
$\pmS{\SIG}{[\SigIS{\SIG[']},\SigIS{\SIG})}
= \pmSp{\SIG}{[\SigIS{\SIG[']},\SigIS{\SIG})} = 0$,
so
$\dS{\SIG}{[\SigIS{\SIG[']},\nfS[{\SIG[']}]{1})} =
\dS{\SIG}{[\SigIS{\SIG},\nfS[{\SIG[']}]{1})}$.
Lemma~\ref{lem:subsequence} implies that
$\nfS[\SIG]{1} \leq \nfS[{\SIG[']}]{1}$, so we can apply
Lemma~\ref{lem:compensate-for-loss} and the induction hypothesis
to show that
\begin{align*}
\dS{\SIG}{[\SigIS{\SIG},\nfS[{\SIG[']}]{1})}
& \stackrel{\text{Lemma~\ref{lem:compensate-for-loss}}}{=}
\dS{\SIG}{[\SigIS{\SIG},\nfS[{\SIG}]{1})}
 \cdot(\inversediffi{\SigIS{\SIG}}{\nfS[{\SIG[']}]{1}})/(\inversediffi{\SigIS{\SIG}}{\nfS[{\SIG}]{1}})
\; \stackrel{\text{I.H.}}{=} \; 0,
\end{align*}
completing the proof.
\end{proof}

We next get to the key lemma for submodularity. It compares the
effects of the increase of $\epsilon$ under $S$ and $T$ on the same
signal \SIG.
At a very high level, it says that the effects on \SIG
in terms of the allocated probability mass of low-value types is more
severe for $T$ than for $S$. The precise form is a bit subtle.
To avoid notational confusion, we will use 
\SIGS to denote signals under the signaling scheme for $S \cup \SET{k}$
and \SigIS{\SIG} to denote their price points;
signals under the signaling scheme for $T \cup \SET{k}$ are denoted by
\SIGT, and their price points by \SigIT{\SIGT}.
Recall that signals are sorted by decreasing price points, 
so that $\SigIS{\SIGS+1} < \SigIS{\SIGS}, \SigIT{\SIGT+1} < \SigIT{\SIGT}$.
We frequently want to find, for a given signal \SIGT, the signal \SIGS
with closest greater (or equal) price point to \SigIT{\SIGT}.
Hence, for any signal \SIGT under $T \cup \SET{k}$, we define
$\ITtoIS{\SIGT} = \max \Set{\SIGS}{\SigIS{\SIGS} \geq \SigIT{\SIGT}}$. 

\begin{lemma} \label{lem:one-signal-comparison}
Let \SIG[k] be the signal with price point $k$ for $S \cup \SET{k}$,
and \SIGT[k] the signal with price point $k$ for $T \cup \SET{k}$. 
Let $\SIGT['] \geq \SIGT[k]$ be any signal, with price point
$k' = \SigIT{\SIGT[']}$. Then,
\begin{align*}
\sum_{\SIG=\SIG[k]}^{\ITtoIS{\SIGT[']}} \dS{\SIG}{[\SigI{\SIG},\nfT[{\SIGT[']}]{j})}
& \leq 
\sum_{\SIGT=\SIGT[k]}^{\SIGT[']} \dT{\SIGT}{[\SigIT{\SIGT},\nfT[{\SIGT[']}]{j})}.
\end{align*}
\end{lemma}

\begin{emptyproof}
We prove the statement by induction on \SIGT['].
The base case $\SIGT['] = \SIGT[k]$ is true because
Lemma~\ref{lem:compensate-for-loss},  applied with
$i = j = k$, $i' = \nfT[{\SIGT[k]}]{j}$ and $j' = \NUMVAL + 1$
implies that
$\dS{\SIGS[k]}{[k, \nfT[{\SIGT[k]}]{j})} 
= \dT{\SIGT[k]}{[k, \nfT[{\SIGT[k]}]{j})} 
= -\epsilon \Val{k} \cdot(\inversediffi{k}{\nfT[{\SIGT[k]}]{j}})$.
We now focus on the induction step, and distinguish three cases.

\begin{enumerate}
\item If $\SigIT{\SIGT[']} \notin S$, then
$\sum_{\SIG=\SIG[k]}^{\ITtoIS{\SIGT[']}} \dS{\SIG}{[\SigI{\SIG},\nfT[{\SIGT[']}]{j})}
= \sum_{\SIG=\SIG[k]}^{\ITtoIS{\SIGT[']-1}} \dS{\SIG}{[\SigI{\SIG},\nfT[{\SIGT[']}]{j})}$.
By induction hypothesis, the latter is at most
$\sum_{\SIGT=\SIGT[k]}^{\SIGT[']-1} \dT{\SIGT}{[\SigIT{\SIGT},\nfT[{\SIGT[']}]{j})}$.
By Lemmas~\ref{lem:compensate-for-loss} and \ref{lem:non-negative-delta},
$\dT{\SIGT[']}{[\SigIT{\SIGT[']}, \nfT[{\SIGT[']}]{j})} \geq 0$, and
adding this inequality proves the induction step.

\item If $\SigIT{\SIGT[']} \in S$ and $\nfS[{\ITtoIS{\SIGT[']}}]{1} \leq k$,
then write $\SIGS['] = \ITtoIS{\SIGT[']}$.
By Lemma~\ref{lem:subsequence}, $\nfT[{\SIGT[']}]{j} = \nfS[{\SIGS[']}]{j'}$ 
for some $j'$, so we can apply Lemma~\ref{lem:compensate-for-loss} 
and Lemma~\ref{lem:non-negative-delta} to conclude that
$\dS{\SIG[']}{[\SigI{\SIG[']},\nfT[{\SIGT[']}]{j})} =
\dS{\SIG[']}{[\SigI{\SIG[']},\nfS[{\SIG[']}]{1})} 
  \cdot(\inversediffi{\SigI{\SIG[']}}{\nfT[{\SIGT[']}]{j}}) 
   / (\inversediffi{\SigI{\SIG[']}}{\nfS[{\SIG[']}]{1}}) =0$. 
As in the previous case, we get that
$\dT{\SIGT[']}{[\SigIT{\SIGT[']}, \nfT[{\SIGT[']}]{j})} \geq 0$,
and adding both terms to the inequality obtained from the induction
hypothesis now completes the inductive step.

\item Otherwise, we are in the case that $\SigIT{\SIGT[']} \in S$
and $\nfS[{\ITtoIS{\SIGT[']}}]{1} > k$; 
we again write $\SIGS['] = \ITtoIS{\SIGT[']}$.
First, we are going to show that the lemma holds for \SIG['] with $j=1$.
By Lemma~\ref{lem:pointwise-dominance}, we obtain that
\begin{align*}
\sum_{\SIGS=\SIGS[k]}^{\SIGS[']} 
     \dS{\SIGS}{[\SigIS{\SIGS},\nfT[{\SIGT[']}]{1})} 
& =  \sum_{\SIGS=\SIGS[k]}^{\SIGS[']}
\big( \pmS{\SIGS}{[\SigIS{\SIGS},\nfT[{\SIGT[']}]{1})} 
    - \pmSp{\SIGS}{[\SigIS{\SIGS},\nfT[{\SIGT[']}]{1})} \big)
\; \leq \; 0,
\end{align*}
whereas by definition of \nfT[{\SIGT[']}]{1},
\begin{align*}
  \sum_{\SIGT=\SIGT[k]}^{\SIGT[']} 
     \dT{\SIGT}{[\SigIT{\SIGT},\nfT[{\SIGT[']}]{1})} 
& = \sum_{\SIGT=\SIGT[k]}^{\SIGT[']}
  \big( \pmT{\SIGT}{[\SigIT{\SIGT},\nfT[{\SIGT[']}]{1})}
  -\pmTp{\SIGT}{[\SigIT{\SIGT},\nfT[{\SIGT[']}]{1})} \big)
\; = \;  \PVal{[\SigIT{\SIGT},\nfT[{\SIGT[']}]{1})}
      - \PVal{[\SigIT{\SIGT},\nfT[{\SIGT[']}]{1})} 
\; = \; 0.
\end{align*}
Thus, we have shown that 
\begin{align*}
\sum_{\SIGS=\SIGS[k]}^{\SIGS[']} \dS{\SIGS}{[\SigIS{\SIGS},\nfT[{\SIGT[']}]{1})}
& \leq \sum_{\SIGT=\SIGT[k]}^{\SIGT[']} \dT{\SIGT}{[\SigIT{\SIGT},\nfT[{\SIGT[']}]{1})}
\; = \; \sum_{\SIGS=\SIGS[k]}^{\SIGS[']}
  \sum_{\SIGT: \SigIS{\SIG} \leq \SigIT{\SIGT} < \SigIS{\SIG-1}, \SigIT{\SIGT} \leq k}
     \dT{\SIGT}{[\SigIT{\SIGT},\nfT[{\SIGT[']}]{1})}.
\end{align*}
The induction hypothesis implies the same inequality with
$\SIGS[''] < \SIGS[']$ in place of \SIGS['], so that we have for all
$\SIGS[''] \leq \SIGS[']$:
\begin{align}
\sum_{\SIGS=\SIGS[k]}^{\SIGS['']} \dS{\SIGS}{[\SigIS{\SIGS},\nfT[{\SIGT[']}]{1})}
& \leq \sum_{\SIGT:\SigIS{\SIG['']}\le \SigIT{\SIGT} \le k} 
            \dT{\SIGT}{[\SigIT{\SIGT},\nfT[{\SIGT[']}]{1})}
\; = \; \sum_{\SIGS=\SIGS[k]}^{\SIGS['']}
  \sum_{\SIGT: \SigIS{\SIG} \leq \SigIT{\SIGT} < \SigIS{\SIG-1}, \SigIT{\SIGT} \leq k}
     \dT{\SIGT}{[\SigIT{\SIGT},\nfT[{\SIGT[']}]{1})}.
\label{eqn:prefix-comparison}
\end{align}

To extend the result to $j > 1$, 
consider any signal $\SIG \in \SET{\SIG[k]+1, \ldots, \SIG[']}$.
By Part (\ref{lem:subsequence:XvsY}) of Lemma~\ref{lem:subsequence}, 
\nfT[{\SIGT[']}]{1} is equal to \nfS[{\SIG[']}]{j} for some $j$.
Because $\nfS[{\SIG[']}]{1} > k \geq \SigI{\SIG}$, 
by Part (\ref{lem:subsequence:differentJ}) of Lemma~\ref{lem:subsequence}, 
both \nfS[{\SIG[']}]{1} and \nfT[{\SIG[']}]{1} occur as \nfS[\SIG]{j},
\nfS[\SIG]{j'} for some $j, j'$.
We are therefore allowed to apply Lemma~\ref{lem:compensate-for-loss},
and we can write
\begin{align*}
  \sum_{\SIG=\SIG[k]}^{\SIG[']} \dS{\SIG}{[\SigI{\SIG},\nfT[{\SIGT[']}]{j})}  
& = \sum_{\SIG=\SIG[k]}^{\SIG[']}
 \dS{\SIG}{[\SigI{\SIG},\nfT[{\SIGT[']}]{1})}
  \cdot(\inversediffi{\SigIS{\SIG}}{\nfT[{\SIGT[']}]{j}})
     / (\inversediffi{\SigIS{\SIG}}{\nfT[{\SIGT[']}]{1}}), \\
\sum_{\SIGT=\SIGT[k]}^{\SIGT[']}
\dT{\SIGT}{[\SigIT{\SIGT},\nfT[{\SIGT[']}]{j})}  
& = \sum_{\SIG=\SIG[k]}^{\SIG[']}  
  \sum_{\SIGT: \SigIS{\SIG} \leq \SigIT{\SIGT} < \SigIS{\SIG-1}, \SigIT{\SIGT} \leq k}
  \dT{\SIGT}{[\SigIT{\SIGT},\nfT[{\SIGT[']}]{1})}
  \cdot (\inversediffi{\SigIT{\SIGT}}{\nfT[{\SIGT[']}]{j}})
     /  (\inversediffi{\SigIT{\SIGT}}{\nfT[{\SIGT[']}]{1}})\\
& \geq \sum_{\SIG=\SIG[k]}^{\SIG[']}  
    (\inversediffi{\SigIS{\SIG}}{\nfT[{\SIGT[']}]{j}})
  / (\inversediffi{\SigIS{\SIGS}}{\nfT[{\SIGT[']}]{1}})
  \cdot \sum_{\SIGT: \SigIS{\SIG} \leq \SigIT{\SIGT} < \SigIS{\SIG-1}, \SigIT{\SIGT} \leq k}
     \dT{\SIGT}{[\SigIT{\SIGT},\nfT[{\SIGT[']}]{1})}.
\end{align*}
Because $\nfT[{\SIGT[']}]{j} \geq \nfT[{\SIGT[']}]{1}$, the function 
$x \mapsto (\frac{1}{x} - \frac{1}{\Val{\nfT[{\SIGT[']}]{j}}}) /
(\frac{1}{x} - \frac{1}{\Val{\nfT[{\SIGT[']}]{1}}})$ is increasing in $x$ for
$x < \Val{\nfT[{\SIGT[']}]{1}}$. 
Furthermore, by Inequality~\eqref{eqn:prefix-comparison},
we have domination of all prefix sums, so we can
apply Lemma~\ref{lem:prefix-sums-imply-combination} 
with $a_{\SIG} = \dS{\SIG}{[\SigI{\SIG},\nfT[{\SIGT[']}]{1})}$,
$b_{\SIG} = \sum_{\SIGT: \SigIS{\SIG} \leq \SigIT{\SIGT} < \SigIS{\SIG-1}, \SigIT{\SIGT} \leq k}
     \dT{\SIGT}{[\SigIT{\SIGT},\nfT[{\SIGT[']}]{1})}$, and
$c_{\SIG} = (\inversediffi{\SigIS{\SIG}}{\nfT[{\SIGT[']}]{j}})
/(\inversediffi{\SigIS{\SIG}}{\nfT[{\SIGT[']}]{1}})$ to conclude that
$\sum_{\SIG=\SIG[k]}^{\SIG[']} \dS{\SIG}{[\SigI{\SIG},\nfT[{\SIGT[']}]{j})}  
\leq \sum_{\SIGT=\SIGT[k]}^{\SIGT[']} \dT{\SIGT}{[\SigIT{\SIGT},\nfT[{\SIGT[']}]{j})}$,
completing the inductive step. \QED
\end{enumerate}
\end{emptyproof}

\begin{lemma} \label{lem:prefix-sums-imply-combination}
Let $a_1, \ldots, a_n$ and $b_1, \ldots, b_n$ be any numbers such that
for all indices $i \leq n$, the prefixes satisfy that
$\sum_{j=1}^i a_j \leq \sum_{j=1}^i b_j$.
Then, for any coefficients $c_1 \geq c_2 \geq \ldots \geq c_n \geq 0$,
we have that
$\sum_{j=1}^n c_j a_j \leq \sum_{j=1}^n c_j b_j$.
\end{lemma}

\begin{proof}
By defining $c_{n+1} = 0$, we can write
$c_j = \sum_{i=j}^n (c_i - c_{i+1})$. 
Now, we get that
\begin{align*}
\sum_{j=1}^n c_j a_j
& = \sum_{j=1}^n \sum_{i=j}^n a_j (c_i - c_{i+1})
\; = \; \sum_{i=1}^n (c_i - c_{i+1}) \sum_{j=1}^i a_j
\; \leq \; \sum_{i=1}^n (c_i - c_{i+1}) \sum_{j=1}^i b_j
& = & \sum_{j=1}^n c_j b_j.
\end{align*}
The inequality followed by the assumption on the prefixes and the
non-negativity of all the $c_i - c_{i+1}$ terms.
\end{proof}

\begin{extraproof}{Lemma~\ref{lem:small-increase-welfare}}
Let \SIG[k] be the signal with price point $k$ to which probability
mass was added.
This addition can only lead to an
increase in social welfare via an increase in buyer types that were
previously not allocated to any signal, meaning that we are interested
only in types $i$ with $\sum_{\SIGS} \pmS{\SIGS}{i} < \PVal{i}$.
In addition to the indices \nfS{j} defined earlier for $j \geq 0$, define
$\nfS{0} > \nfS{-1} > \nfS{-2} > \dots$ to be all of the indices $i < \nfS{0}$ 
with $\pmS{\SIGS}{\SIGSETS} < \PVal{i}$.

We will first show that those types $i < \nfS{0}$ do not actually lead
to any welfare changes under \PMS.
Thereto, consider some type \nfS{j} with $j < 0$.
Under \PMS, buyers of this type can only be allocated to a signal \SIGS with
$\SigIS{\SIGS} \leq \nfS{j}$. 
For all such signals \SIGS, by definition, we have that
$\nfS[{\SIGS}]{1} \leq \nfS{j} < \nfS{0} \leq k$, so
Lemma~\ref{lem:non-negative-delta} implies that
$\dS{\SIGS}{[\SigIS{\SIGS},\nfS[{\SIGS}]{1})} = 0$. 
By Part (\ref{lem:subsequence:differentJ}) of Lemma~\ref{lem:subsequence},
\nfS{j} and \nfS{j+1} also occur as \nfS[\SIGS]{j'}, \nfS[\SIGS]{j''},
so we can apply Lemma~\ref{lem:compensate-for-loss} to conclude that
$\dS{\SIGS}{[\nfS{j}, \nfS{j+1})} = 0$, and by summing obtain that for
all $i < 0$,
\begin{align*}
\sum_{\SIGS} \dS{\SIGS}{\nfS{j}}
& \stackrel{(*)}{=} \sum_{\SIGS} \dS{\SIGS}{[\nfS{j}, \nfS{j+1})}
\; = \; \sum_{\SIGS: \SigIS{\SIGS} \leq \nfS{j}} \dS{\SIGS}{[\nfS{j},\nfS{j+1})}
\; = \; 0
\end{align*}
In the step labeled (*), we are using the fact that 
\nfS{i} is the only index in the range $[\nfS{i},\nfS{i+1})$ at which
the probability mass can increase, and the probability mass for no
other type will decrease. 

By Part~\ref{lem:subsequence:unindexed} of Lemma~\ref{lem:subsequence},
the \nfT{j} form a subsequence of the \nfS{j}.
To compare the effects of changes in the welfare, 
we partition the \nfS{j} into the segments formed by the \nfT{j};
formally, we define $\StoT{\nfS{j}} = \max \Set{\nfT{j'}}{\nfT{j'} \leq \nfS{j}}$.
Let \START be the index $j$ such that $\nfS{\START} = \nfT{1}$.

We can now express the change in social welfare when adding $k$ to $S$
as follows:
\begin{align}
\PWelfare{\PMS} - \PWelfare{\PMSP}
& = \sum_{j} \Val{\nfS{j}} \sum_{\SIGS \geq \SIG[k]} \dS{\SIGS}{\nfS{j}} 
\nonumber \\
& = \sum_{j \geq \START} \Val{\nfS{j}} \cdot \sum_{\SIGS \geq \SIGS[k]} \dS{\SIGS}{[\nfS{j},\nfS{j+1})}
+  \sum_{j=0}^{\START-1} \Val{\nfS{i}} \sum_{\SIGS \geq \SIGS[k]} \dS{\SIGS}{\nfS{j}} 
\nonumber \\ 
& \stackrel{(*)}{=} \sum_{j \geq \START} \Val{\nfS{j}} \cdot 
    \sum_{\SIGS \geq \SIGS[k]} \dS{\SIGS}{[\SigIS{\SIGS},\nfS{1})}
  \cdot (\inversediffi{\nfS{j}}{\nfS{j+1}}) / (\inversediffi{\SigIS{\SIGS}}{\nfS{1}})
+  \sum_{j=0}^{\START-1} \Val{\nfS{i}} \sum_{\SIGS \geq \SIGS[k]} \dS{\SIGS}{\nfS{j}} 
\nonumber \\ 
& = \big( \sum_{j \geq \START} \Val{\nfS{j}}
  \cdot (\inversediffi{\nfS{j}}{\nfS{j+1}}) \big) \cdot
  \big( \sum_{\SIGS \geq \SIGS[k]} \dS{\SIGS}{[\SigIS{\SIGS},\nfS{1})}
  /(\inversediffi{\SigIS{\SIGS}}{\nfS{1}}) \big)
+  \sum_{j=0}^{\START-1} \Val{\nfS{i}} \sum_{\SIGS \geq \SIGS[k]} \dS{\SIGS}{\nfS{j}}.
\label{eqn:wefare-difference}
\end{align}
In the step labeled (*), we applied Lemma~\ref{lem:compensate-for-loss}.
We were allowed to do so, because 
$\nfS{j+1} \geq \nfS{j} \geq \nfT{1} \geq \nfS{1} \geq k \geq \SigIS{\SIG}$,
allowing us to apply 
Part (\ref{lem:subsequence:differentJ}) of Lemma~\ref{lem:subsequence}.

We first analyze the second term 
$\sum_{j=0}^{\START-1} \Val{\nfS{i}} \sum_{\SIGS \geq \SIGS[k]} \dS{\SIGS}{\nfS{j}}$.
If \nfT{0} is defined, then
\begin{align*}
\sum_{j=0}^{\START-1} \Val{\nfS{j}} \sum_{\SIGS \geq \SIGS[k]} \dS{\SIGS}{\nfS{j}}
& \leq \Val{\nfT{0}} \cdot 
  \sum_{\SIGS \geq \SIGS[k]} \sum_{j=1}^{\START-1} \dS{\SIGS}{\nfS{j}}
\; = \; \Val{\nfT{0}} \cdot 
  \sum_{\SIGS \geq \SIGS[k]} \dS{\SIGS}{[\SigIS{\SIGS},\nfT{1})}
\; \stackrel{\text{Lemma~\ref{lem:one-signal-comparison}}}{\leq} \; 
  \Val{\nfT{0}} \cdot \sum_{\SIGT \geq \SIGT[k]} \dT{\SIGT}{[\SigIT{\SIGT},\nfT{1})}.
\end{align*}
Otherwise (\nfT{0} is not defined), we will simply use the bound that
$\sum_{\SIGS \geq \SIGS[k]} \dS{\SIGS}{[\SigIS{\SIGS},\nfS{1})} \leq 0$
by Lemma~\ref{lem:pointwise-dominance}.
For notational convenience, in this case, we will write $\Val{\nfT{0}} = 0$.

We next consider the first factor of the first term of
Equation~\eqref{eqn:wefare-difference}. 
Because $\nfS{j} \geq \StoT{\nfS{j}}$, we obtain that
\begin{align*}
\sum_{j \geq \START} \Val{\nfS{j}} \cdot (\inversediffi{\nfS{j}}{\nfS{j+1}})
& \geq
\sum_{j \geq \START} \Val{\StoT{\nfS{j}}} \cdot (\inversediffi{\nfS{j}}{\nfS{j+1}})
\; = \; \sum_{j' \geq 1} \Val{\nfT{j'}} \cdot
\sum_{j: \StoT{\nfS{j}} = \nfT{j'}} (\inversediffi{\nfS{j}}{\nfS{j+1}})\\
& = \sum_{j' \geq 1} \Val{\nfT{j'}} \cdot (\inversediffi{\nfT{j'}}{\nfT{j'+1}}).
\end{align*}

To analyze the second factor of the first term in 
Equation~\eqref{eqn:wefare-difference}, we first apply
Lemma~\ref{lem:compensate-for-loss} to rewrite
$\sum_{\SIGS \geq \SIGS[k]} 
  \dS{\SIGS}{[\SigIS{\SIGS},\nfS{1})} / (\inversediffi{\SigIS{\SIGS}}{\nfS{1}})
= \sum_{\SIGS \geq \SIGS[k]} 
  \dS{\SIGS}{[\SigIS{\SIGS},\nfT{1})} / (\inversediffi{\SigIS{\SIGS}}{\nfT{1}})$.
We also write
\begin{align*}
\sum_{\SIGT \geq \SIGT[k]} \dT{\SIGT}{[\SigIT{\SIGT},\nfT{1})}
    /(\inversediffi{\SigIT{\SIGT}}{\nfT{1}})
& = \sum_{\SIGS \geq \SIGS[k]}   \sum_{\SIGT: \ITtoISSUM{\SIGS}}
       \dT{\SIGT}{[\SigIT{\SIGT},\nfT{1})}
     /(\inversediffi{\SigIT{\SIGT}}{\nfT{1}})
\;\\ &\geq \; \sum_{\SIGS \geq \SIGS[k]} 
     1/(\inversediffi{\SigIS{\SIGS}}{\nfT{1}})
    \sum_{\SIGT: \ITtoISSUM{\SIGS}} \dT{\SIGT}{[\SigIT{\SIGT},\nfT{1})}.
\end{align*}
By Lemma~\ref{lem:one-signal-comparison}, we have that
$\sum_{\SIG=\SIG[k]}^{\ITtoIS{\SIGT[']}} \dS{\SIG}{[\SigI{\SIG},\nfT[{\SIGT[']}]{j})}
\leq \sum_{\SIGT=\SIGT[k]}^{\SIGT[']} \dT{\SIGT}{[\SigIT{\SIGT},\nfT[{\SIGT[']}]{j})}$
for all $\SIGT['] \geq \SIGT[k]$.

We can therefore apply Lemma~\ref{lem:prefix-sums-imply-combination}
with $a_{\SIGS} = \dS{\SIG}{[\SigI{\SIG},\nfT[{\SIGT[']}]{1})}$,
$b_{\SIGS} = \sum_{\SIGT: \ITtoISSUM{\SIGS}} \dT{\SIGT}{[\SigIT{\SIGT},\nfT{1})}$,
and $c_{\SIGS} = 1/(\inversediffi{\SigIS{\SIGS}}{\nfT{1}})$, and
conclude that
\begin{align*}
\sum_{\SIGS \geq \SIGS[k]} 
  \dS{\SIGS}{[\SigIS{\SIGS},\nfS{1})} / (\inversediffi{\SigIS{\SIGS}}{\nfS{1}})
& = \sum_{\SIGS \geq \SIGS[k]} 
  \dS{\SIGS}{[\SigIS{\SIGS},\nfT{1})} / (\inversediffi{\SigIS{\SIGS}}{\nfT{1}})
\; \leq \; \sum_{\SIGS \geq \SIGS[k]} 
     1/(\inversediffi{\SigIS{\SIGS}}{\nfT{1}}) \cdot
  \sum_{\SIGT: \ITtoISSUM{\SIGS}} \dT{\SIGT}{[\SigIT{\SIGT},\nfT{1})}\\
& \leq \sum_{\SIGT \geq \SIGT[k]} \dT{\SIGT}{[\SigIT{\SIGT},\nfT{1})}
    /(\inversediffi{\SigIT{\SIGT}}{\nfT{1}}).
\end{align*}

Recalling that 
$\sum_{\SIGS \geq \SIGS[k]} 
  \dS{\SIGS}{[\SigIS{\SIGS},\nfS{1})} / (\inversediffi{\SigIS{\SIGS}}{\nfS{1}})
\leq 0$ by Lemma~\ref{lem:pointwise-dominance}, we now put all of
these inequalities together, yielding that
\begin{align*}
\PWelfare{\PMS} - \PWelfare{\PMSP} 
& \leq 
\big(\sum_{j' \geq 1} \Val{\nfT{j'}} \cdot (\inversediffi{\nfT{j'}}{\nfT{j'+1}}) \big)
\cdot 
\big( \sum_{\SIGT \geq \SIGT[k]} \dT{\SIGT}{[\SigIT{\SIGT},\nfT{1})}
    / (\inversediffi{\SigIT{\SIGT}}{\nfT{1}}) \big)
+ \Val{\nfT{0}} \cdot \sum_{\SIGT \geq \SIGT[k]} \dT{\SIGT}{[\SigIT{\SIGT},\nfT{1})} \\
& = 
\sum_{j' \geq 1} \Val{\nfT{j'}} \cdot 
\sum_{\SIGT \geq \SIGT[k]} \dT{\SIGT}{[\SigIT{\SIGT},\nfT{1})}
\cdot (\inversediffi{\nfT{j'}}{\nfT{j'+1}})
    / (\inversediffi{\SigIT{\SIGT}}{\nfT{1}})
+ \Val{\nfT{0}} \cdot \sum_{\SIGT \geq \SIGT[k]} \dT{\SIGT}{[\SigIT{\SIGT},\nfT{1})}\\
& \stackrel{\text{Lemma~\ref{lem:compensate-for-loss}}}{=}
\sum_{j' \geq 1} \Val{\nfT{j'}} \cdot 
\sum_{\SIGT \geq \SIGT[k]} \dT{\SIGT}{[\nfT{j'}, \nfT{j'+1})}
+ \Val{\nfT{0}} \cdot \sum_{\SIGT \geq \SIGT[k]} \dT{\SIGT}{[\SigIT{\SIGT},\nfT{1})} \\
& = \sum_{j' \geq 0} \Val{\nfT{j'}} \cdot 
 \sum_{\SIGT \geq \SIGT[k]} \dT{\SIGT}{\nfT{j'}}\\
& =  \PWelfare{\PMT} - \PWelfare{\PMTP}.
\end{align*}

Finally, by noticing that the respective increases in \pwf
are the negatives of the terms here 
(i.e., $\PWelfare{\PMSP} - \PWelfare{\PMS}$ 
and $\PWelfare{\PMTP} - \PWelfare{\PMT}$, respectively), we complete the
proof of Lemma~\ref{lem:small-increase-welfare}.
\end{extraproof}

\end{document}